\documentclass[a4paper,UKenglish]{lipics-v2019}

\usepackage{microtype}%
\hideLIPIcs
\title{Efficient Algorithms for Ortho-Radial Graph Drawing}

\titlerunning{Efficient Algorithms for Ortho-Radial Graph Drawing}%

\author{Benjamin Niedermann}{University of Bonn}{niedermann@uni-bonn.de}{https://orcid.org/0000-0001-6638-7250}{}%

\author{Ignaz Rutter}{University of 
Passau}{rutter@fim.uni-passau.de}{https://orcid.org/0000-0002-3794-4406}{}

\author{Matthias Wolf}{Karlsruhe Institute of 
Technology}{matthias.wolf@kit.edu}{https://orcid.org/0000-0003-1411-6330}{}

\authorrunning{B. Niedermann, I. Rutter, and M. Wolf}%

\Copyright{Benjamin Niedermann, Ignaz Rutter, and Matthias Wolf}%



\category{}%

\relatedversion{}%

\supplement{}%

\funding{}%

\EventEditors{John Q. Open and Joan R. Access}
\EventNoEds{2}
\EventLongTitle{}
\EventShortTitle{arXiv 2019}
\EventAcronym{arXiv}
\EventYear{2019}
\EventDate{}
\EventLocation{}
\EventLogo{}
\SeriesVolume{}
\ArticleNo{1}
\nolinenumbers %
\hideLIPIcs  %

\usepackage{amsmath}
\usepackage{amsthm}
\usepackage{color}
\usepackage{gensymb}
\usepackage{todonotes}
\usepackage{hyperref}
\usepackage{sidecap}
\usepackage{wrapfig}
\usepackage{paralist}
\usepackage{thmtools}
\usepackage{thm-restate}
\usepackage{subcaption}
\usepackage[numbers,sort]{natbib}

\theoremstyle{definition}

\declaretheorem[name=Theorem,numberlike=theorem,style=plain]{restatable-theorem}
\declaretheorem[name=Lemma,numberlike=theorem,style=plain]{restatable-lemma}

\newcommand{\widebar}[1]{\mkern 
	1.5mu\overline{\mkern-1.5mu#1\mkern-1.5mu}\mkern 1.5mu}
\newcommand{\subpath}[2]{\ensuremath{#1[#2]}}

\newcommand{\reverse}[1]{\ensuremath{\widebar{#1}}}

\renewcommand{\O}{\ensuremath{O}}
\newcommand{\flip}[1]{\reverse{#1}}
\newcommand{\mirror}[1]{\hat{#1}}

\DeclareMathOperator{\rot}{rot}

\newcommand{\T}{\mathrm{T}}
\newcommand{\B}{\mathrm{B}}

\begin{document}

\maketitle

\begin{abstract}
  Orthogonal drawings, i.e., embeddings of graphs into grids, are a
  classic topic in Graph Drawing.  Often the goal is to find a drawing
  that minimizes the number of bends on the edges.  A key ingredient
  for bend minimization algorithms is the existence of an
  \emph{orthogonal representation} that allows to describe such
  drawings purely combinatorially by only listing the angles between
  the edges around each vertex and the directions of bends on the
  edges, but neglecting any kind of geometric information such as
  vertex coordinates or edge lengths.

  Barth et al.~\cite{bnrw-ttsmford-17} have established the existence
  of an analogous \emph{ortho-radial representation} for
  \emph{ortho-radial drawings}, which are embeddings into an
  ortho-radial grid, whose gridlines are concentric circles around the
  origin and straight-line spokes emanating from the origin but
  excluding the origin itself. While any orthogonal representation
  admits an orthogonal drawing, it is the circularity of the
  ortho-radial grid that makes the problem of characterizing valid
  ortho-radial representations all the more complex and interesting.
  Barth et al.~prove such a characterization.  However, the proof is
  existential and does not provide an efficient algorithm for testing
  whether a given ortho-radial representation is valid, let alone
  actually obtaining a drawing from an ortho-radial representation.

  In this paper we give quadratic-time algorithms for both of these
  tasks.  They are based on a suitably constrained left-first DFS in
  planar graphs and several new insights on ortho-radial
  representations.  Our validity check requires quadratic time, and a
  naive application of it would yield a quartic algorithm for
  constructing a drawing from a valid ortho-radial
  representation. Using further structural insights we speed up the
  drawing algorithm to quadratic running time.
\end{abstract}

\section{Introduction}
\label{sec:introduction}

Grid drawings of graphs embed graphs into grids such that vertices map
to grid points and edges map to internally disjoint curves on the grid
lines that connect their endpoints.  Orthogonal grids, whose grid
lines are horizontal and vertical lines, are popular and widely used
in graph drawing. Among others, orthogonal graph drawings are applied
in VLSI design (e.g.,~\cite{Valiant1981,Bhatt1984}), diagrams
(e.g.,~\cite{Batini1986,Gutwenger2003,Eiglsperger2004,Wybrow2010}), and
network layouts (e.g.,~\cite{Ruegg2014,Kieffer2016}).  They have been extensively studied
with respect to their construction and properties
(e.g.,~\cite{Tamassia1991,Biedl1996,Biedl1998,Papakostas1998,Alam2017}).  
Moreover, they
have been generalized to arbitrary planar graphs with degree higher
than four (e.g.,~\cite{Tamassia1988,Fossmeier1996,Biedl1997}).

\begin{figure}[t]
\begin{minipage}[b]{0.52\textwidth}
  \centering
  \begin{subfigure}[b]{0.48\textwidth}
    \centering
    \includegraphics{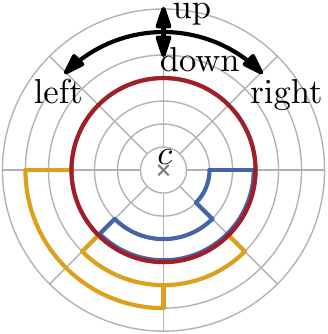}
    \caption{Ortho-radial grid.}
    \label{fig:pre:drawing-grid}
  \end{subfigure}
  \hfill
  \begin{subfigure}[b]{0.44\textwidth}
    \centering
    \includegraphics{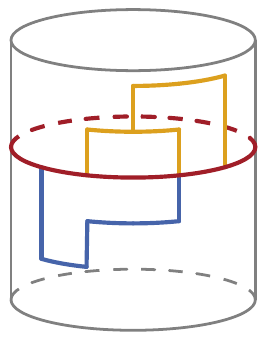}
    \caption{Cylinder drawing.}
    \label{fig:pre:drawing-cylinder}
  \end{subfigure}
  \hfill
  \caption{An ortho-radial drawing of a graph on a grid
    \protect(\subref{fig:pre:drawing-grid}) and its equivalent interpretation
    as an orthogonal drawing on a cylinder
    \protect(\subref{fig:pre:drawing-cylinder}).}
  \label{fig:pre:drawing}
\end{minipage}
\hfill
 \begin{minipage}[b]{0.46\textwidth}
   \centering
   \includegraphics[width=.9\textwidth]{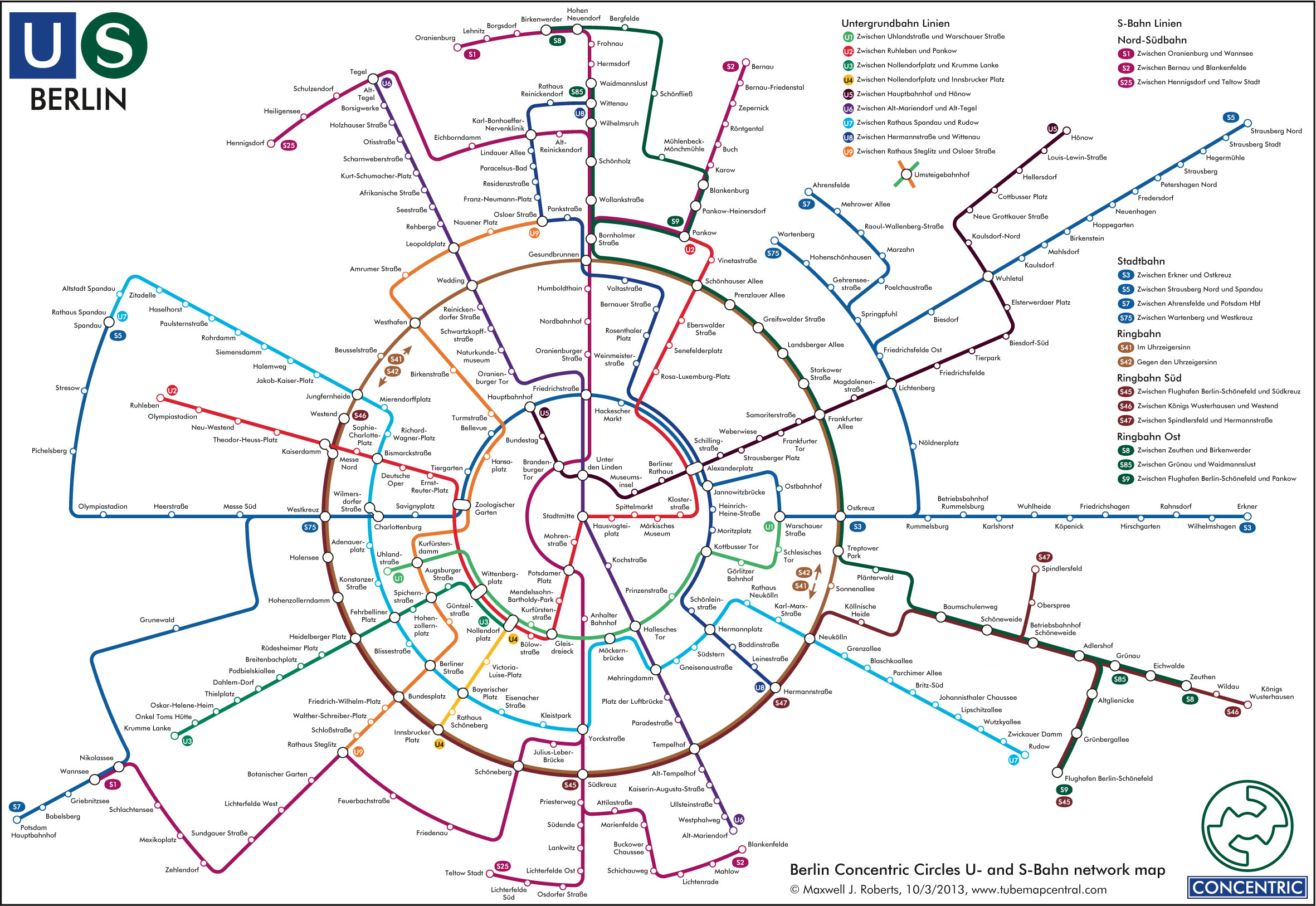}
   
   \caption{Metro map of Berlin using an ortho-radial 
   layout\protect\footnotemark{}.
   Image
     copyright by Maxwell J.~Roberts. Reproduced with permission.}
   \label{fig:intro:berlin}
 \end{minipage} 
\end{figure}
\footnotetext{Note that ortho-radial
  drawings exclude the center of the grid, which is slightly different
  to the concentric circles maps by Maxwell J.\ Roberts.}

Ortho-radial drawings are a generalization of orthogonal drawings to
grids that are formed by concentric circles and straight-line spokes
from the center but excluding the center.  Equivalently, they can be viewed as graphs drawn in an orthogonal fashion on the
surface of a standing cylinder, see Figure~\ref{fig:pre:drawing}, or a
sphere without poles. Hence, they naturally bring orthogonal graph
drawings to the third dimension.  

Among other applications, ortho-radial drawings are used to
  visualize network maps; see Figure~\ref{fig:intro:berlin}.
  Especially, for metro systems of metropolitan areas they are highly
  suitable. Their inherent structure emphasizes the city center, the
  metro lines that run in circles as well as the metro lines that lead to
  suburban areas. While the automatic creation of metro maps has been
  extensively studied for other layout
  styles~(e.g.,~\cite{Hong2006,Noellenburg2011,Wang2011,Fink2013}),
  this is a new and wide research field for ortho-radial drawings.

Adapting existing techniques and objectives from orthogonal
  graph drawings is a promising step to open up that
  field.  One main objective in orthogonal
graph drawing is to minimize the number of bends on the edges.  The
core of a large fraction of the algorithmic work on this problem is
the \emph{orthogonal representation}, introduced by
Tamassia~\cite{t-emn-87}, which describes orthogonal drawings listing
\begin{inparaenum}[(i)]
\item the angles formed by consecutive edges around each vertex and
\item the directions of bends along the edges.
\end{inparaenum}
Such a representation is \emph{valid} if
\begin{inparaenum}[(I)]
\item the angles around each vertex sum to $360\degree$, and
\item the sum of the angles around each face with $k$ vertices is $(k-2)\cdot 
180\degree$ for internal faces and $(k+2)\cdot 180\degree$ for the outer face.
\end{inparaenum}
The necessity of the first condition is obvious and the necessity of
the latter follows from the sum of inner/outer angles of any polygon
with $k$ corners.  It is thus clear that any orthogonal drawing yields
a valid orthogonal representation, and Tamassia~\cite{t-emn-87} showed
that the converse holds true as well; for a valid orthogonal
representation there exists a corresponding orthogonal drawing that
realizes this representation.  Moreover, the proof is constructive and
allows the efficient construction of such a drawing, a process that is
referred to as \emph{compaction}.

Altogether this enables a three-step approach for computing orthogonal
drawings, the so-called \emph{Topology-Shape-Metrics Framework}, which
works as follows.  First, fix a \emph{topology}, i.e., combinatorial
embedding of the graph in the plane (possibly planarizing it if it is
non-planar); second, determine the \emph{shape} of the drawing by
constructing a valid orthogonal representation with few bends; and
finally, compactify the orthogonal representation by assigning
suitable vertex coordinates and edge lengths (\emph{metrics}).  As
mentioned before, this reduces the problem of computing an orthogonal
drawing of a planar graph with a fixed embedding to the purely
combinatorial problem of finding a valid orthogonal representation,
preferably with few bends.  The task of actually creating a
corresponding drawing in polynomial time is then taken over by the
framework.  It is this approach that is at the heart of a large body
of literature on bend minimization algorithms for orthogonal
drawings (e.g., \cite{Bertolazzi2000,Eiglsperger2003,Cornelsen2012,Felsner2014,%
Blasius2016,Blasius2016b,Chang2017}).

Very recently Barth et al.~\cite{bnrw-ttsmford-17} proposed a
generalization of orthogonal representations to ortho-radial drawings,
called \emph{ortho-radial} representations, with the goal of
establishing an ortho-radial analogue of the TSM framework for
ortho-radial drawings. They show that a natural generalization of the
validity conditions (I) and (II) above is not sufficient, and
introduce a third, less local condition that excludes so-called
\emph{monotone cycles}, which do not admit an ortho-radial drawing.
They show that these three conditions together fully characterize
ortho-radial drawings.
Before that, characterizations for bend-free ortho-radial drawings
were only known for paths, cycles and theta
graphs~\cite{hht-orthoradial-09}. Further, for the special case that
each internal face is a rectangle, a characterization for cubic graphs
was known~\cite{hhmt-rrdcp-10}.

With the result by Barth et al.~finding an ortho-radial
  drawing for a planar graph with fixed-embedding reduces to the
  purely combinatorial problem of finding a valid ortho-radial
  representation. In particular, since bends can be seen as
  additionally introduced vertices subdividing edges, finding an
  ortho-radial drawing with minimum number of bends reduces to  finding a valid ortho-radial
  representation with minimum number of such additionally introduced vertices. In this sense, the work by Barth et
  al.~constitutes a major step towards computing ortho-radial drawings
  with minimum number of bends.

Yet, it is here where their work still contains a major gap. While
the work of Barth et al.\ shows that valid ortho-radial
representations fully characterize ortho-radial drawings, it is
unclear if it can be checked efficiently whether a given ortho-radial
representation is valid.  Moreover, while their existential proof of a
corresponding drawing is constructive, it needs to repeatedly test
whether certain ortho-radial representations are valid.

\subparagraph*{Contribution and Outline.} We develop such a
test running in quadratic time, thus implementing the compaction step
of the TSM framework with polynomial running time.  While this does
not yet directly allow us to compute ortho-radial drawings with few bends, our
result paves the way for a purely combinatorial treatment of bend
minimization in ortho-radial drawings, thus enabling the same type of
tools that have proven highly successful in minimizing bends in
orthogonal drawings. 

At the core of our validity testing algorithm are several new insights
into the structure of ortho-radial representations.  The algorithm
itself is a left-first DFS that uses suitable constraints to determine
candidates for monotone cycles in such a way that if a given
ortho-radial representation contains a monotone cycle, then one of the
candidates is monotone. While it may be obvious to use a DFS for
  finding cycles in general, it is far from clear how such a search
  works for monotone cycles in ortho-radial representations.
Plugging this test as a black box into the drawing algorithm of Barth
et al.\ yields an $\O(n^{4})$-time algorithm for computing a drawing
from a valid ortho-radial representation, where $n$ is the number of
vertices.  Using further structural insights on the augmentation
process we improve the running time of this algorithm to
$\O(n^{2})$. Hence, our result is not only of theoretical
  interest, but the algorithm can be actually deployed. We believe
  that the algorithm is a useful intermediate step for providing
  initial network layouts to map designers and layout algorithms such
  as force directed algorithms; see also
  Section~\ref{sec:conclusion}.

In Section~\ref{sec:preliminaries} we present preliminaries that are
used throughout the paper.  First we formally define ortho-radial
representations and recall the most important results
from~\cite{bnrw-ttsmford-17}.  Afterwards, in Section~\ref{sec:symmetries-normalization-appendix}, we show that for the purpose
of validity checking and determining the existence of a monotone
cycle, we can restrict ourselves to so-called \emph{normalized
  instances}. In Section~\ref{sec:finding_monotone_cycles} we give a
validity test for ortho-radial representations that runs in
$\O(n^{2})$ time.  Afterwards, in Section~\ref{sec:rectangulation}, we
revisit the rectangulation procedure from~\cite{bnrw-ttsmford-17} and
show that using the techniques from
Section~\ref{sec:finding_monotone_cycles} it can be implemented to run
in $\O(n^{2})$ time, improving over a naive application which would
yield running time $\O(n^{4})$.  Together with~\cite{bnrw-ttsmford-17}
this enables a purely combinatorial treatment of ortho-radial
drawings.  We conclude with a summary and some open questions in
Section~\ref{sec:conclusion}.

\section{Preliminaries}
\label{sec:preliminaries}

We first formally introduce ortho-radial drawings and ortho-radial
representations.  Afterwards we present two transformations that we
use to simplify the discussion of symmetric cases.

\subsection{Ortho-Radial Drawings and Representations}
We use the same definitions and conventions on ortho-radial drawings
as presented by Barth et al.~\cite{bnrw-ttsmford-17}; for the
convenience of the reader we briefly repeat them here.  In particular,
we only consider drawings and representations without bends on the
edges.  As argued in~\cite{bnrw-ttsmford-17}, this is not a
restriction, since it is always possible to transform a
drawing/representation with bends into one without bends by
subdividing edges so that a vertex is placed at each bend.

We are given a planar 4-graph $G=(V,E)$ with $n$ vertices and fixed
embedding, where a graph is a 4-graph if it has only
vertices with degree at most four.  We define that a path $P$ in $G$
is always simple, while a cycle $C$ may contain vertices multiple
times but may not cross itself.
 All cycles are oriented clockwise, so that 
their interiors
are locally to the right.  A cycle is part of its interior and
exterior.  We denote the subpath of $P$ from $u$ to $v$ by $P[u,v]$
assuming that $u$ and $v$ are included.  For any path
$P=v_1,\dots,v_k$ its \emph{reverse} is %
$\reverse{P}=v_k,\dots,v_1$. The concatenation of two paths $P_1$ and
$P_2$ is written as $P_1+P_2$. For a cycle~$C$ in $G$ that contains
any edge at most once, the subpath $C[e,e']$ between two edges $e$ and
$e'$ on $C$ is the unique path on $C$ that starts with $e$ and ends
with $e'$. If the start vertex~$u$ of $e$ is contained in $C$ only
once, we also write $C[u,e']$, because then $e$ is uniquely
defined by $u$. Similarly, if the end vertex $v$ of $e'$ is contained
in $C$ only once, we also write $C[e,v]$.
We also use this notation to refer to subpaths of simple paths.

In an ortho-radial drawing~$\Delta$ of $G$ each edge is directed and drawn
either clockwise, counter-clockwise, towards the center or away from
the center.  Hence, using the metaphor of a cylinder, the edges
point \emph{right}, \emph{left}, \emph{down} or \emph{up},
respectively. Moreover, \emph{horizontal edges} point left or right, while
\emph{vertical edges} point up or down; see Figure~\ref{fig:pre:drawing}.

We distinguish two types of simple cycles. If the center of the grid
lies in the interior of a simple cycle, the cycle is \emph{essential}
and otherwise \emph{non-essential}. Further, there is an unbounded
face in $\Delta$ and a face that contains the center of the grid; we
call the former the \emph{outer face} and the latter the \emph{central
  face}; in our drawings we mark the central face using a small ``x''. All 
  other faces are \emph{regular}.
  
For two edges $uv$ and $vw$ incident to the same vertex $v$, we define
the \emph{rotation} $\rot(uvw)$ as $1$ if there is
a right turn at $v$, $0$ if $uvw$ is straight and $-1$ if there is a
left turn at $v$. In the special case that $u=w$, we have
$\rot(uvw)=-2$.

The rotation of a path $P=v_1,\dots,v_k$ is the sum of the rotations
at its internal vertices, i.e.,
$\sum_{i=2}^{k-1}\rot(v_{i-1}v_iv_{i+1})$.  Similarly, for a cycle
$C = v_1,\ldots,v_k,v_1$, its rotation is the sum of the rotations at
all its vertices (where we define $v_0 = v_k$ and $v_{k+1}=v_1$), i.e.,
$\rot(C) = \sum_{i=1}^{k} \rot(v_{i-1}v_{i}v_{i+1})$.  
We observe that $\rot(P)=\rot(P[s,e])+\rot(P[e,t])$ for any path~$P$
from $s$ to $t$ and any edge $e$ on $P$. Further, we have
$\rot(\reverse{P})=-\rot(P)$.  For a face $f$ we use $\rot(f)$ to denote the 
rotation of the facial cycle that bounds $f$ (oriented such that $f$ lies on 
the right side of the cycle).

As introduced by Barth et al.~\cite{bnrw-ttsmford-17}, an \emph{ortho-radial 
  representation} $\Gamma$ of a $4$-planar graph
$G$ fixes the central and outer face of $G$ as well as a
reference edge $e^\star$ on the outer face such that the outer face is
locally to the left of $e^\star$.  Following the convention established by 
Barth et al.~\cite{bnrw-ttsmford-17} the reference edge always points
right. Further, $\Gamma$ specifies for each face $f$ of $G$ a list $H(f)$
that contains for each edge~$e$ of $f$ the pair $(e,a)$, where
$a\in \{90\degree,180\degree,270\degree,360\degree\}$. The interpretation of 
$(e,a)$ is that the
edge~$e$ is directed such that the interior of $f$ locally lies to the
right of $e$ and $a$ specifies the angle inside $f$ from $e$ to
the following edge.
The notion of rotations can be extended to these descriptions since we can 
compute the angle at a vertex $v$ enclosed by edges $uv$ and $vw$ by summing 
the corresponding angles in the faces given by the $a$-values.
For such a description to be an ortho-radial 
representation, two local conditions need to be satisfied:

\begin{compactenum}
  \item\label{cond:repr:sum_of_angles} The angle sum of all edges around each 
  vertex given by the $a$-fields is 360\degree.
  \item\label{cond:repr:rotation_faces} For each face $f$, we have
  \[
  \rot(f)=
  \begin{cases}
  4, & \text{$f$ is a regular face} \\
  0, & \text{$f$ is the outer or the central face but not both} \\
  -4,& \text{$f$ is both the outer and the central face.} \\
  \end{cases}
  \]
\end{compactenum}

These conditions ensure that angles are assigned correctly around vertices and 
inside faces, which implies that all properties of rotations mentioned above 
hold. 
An ortho-radial representation $\Gamma$ of a graph $G$ is \emph{drawable} if 
there is a drawing $\Delta$ of $G$ embedded as specified by $\Gamma$ such that 
the corresponding angles in $\Delta$ and $\Gamma$ are equal and the reference 
edge $e^\star$ points to the right.
Unlike for orthogonal representations the two conditions do not guarantee 
that 
the ortho-radial representation is drawable. Therefore, Barth et 
al.~\cite{bnrw-ttsmford-17} introduced a third condition, which is formulated 
in terms of labelings of essential cycles.

For a simple, essential cycle $C$ in $G$
and a path $P$ from the target vertex~$s$ of the reference edge $e^\star$ to 
a 
vertex $v$ on $C$ the \emph{labeling} $\ell^P_C$ assigns to each edge $e$ on 
$C$ the label $\ell^P_C(e)=\rot(e^\star+P+C[v,e])$. In this paper we 
always assume that $P$ is \emph{elementary}, i.e., $P$ intersects $C$ only
at its endpoints. For these paths the labeling is independent of the actual 
choice of $P$, which was shown by Barth et al.~\cite{bnrw-ttsmford-17}. We 
therefore drop the superscript $P$ and write $\ell_C(e)$ for the labeling of an
edge $e$ on an essential cycle~$C$.
We call an essential cycle
\emph{monotone} if either all its labels are non-negative or all its
labels are non-positive. A monotone cycle is a \emph{decreasing} cycle
if it has at least one strictly positive label, and it is an \emph{increasing} 
cycle if $C$ has at least one strictly
negative label.  An ortho-radial representation is \emph{valid} if it
contains neither decreasing nor increasing cycles.  The validity of an
ortho-radial representation ensures that on each essential cycle with
at least one non-zero label there is at least one edge pointing up and
one pointing down.
The main theorem of Barth et al.~\cite{bnrw-ttsmford-17} can be stated as 
follows.\footnote{In the following we refer to the full  
version~\cite{bnrw-ttsmford-17-arxiv} of \cite{bnrw-ttsmford-17}, when citing 
lemmas and theorems.}

\begin{proposition}[Reformulation of Theorem~5 
  in~\cite{bnrw-ttsmford-17-arxiv}]\label{prop:characterization}
 An ortho-radial representation is drawable if and only if it is valid.
\end{proposition}

To that end, Barth et al.~\cite{bnrw-ttsmford-17-arxiv} prove the
following results among others. Since we use them throughout this
paper, we restate them for the convenience of the reader.
Both assume ortho-radial representations that are not necessarily valid.

\begin{proposition}[Lemma~12 in 
\cite{bnrw-ttsmford-17-arxiv}]\label{lem:repr:equal_labels_at_intersection}
Let $C_1$ and $C_2$ be two essential cycles and let $H=C_1+C_2$ be the
subgraph of $G$ formed by these two cycles. For any common edge $vw$
of $C_1$ and $C_2$ where $v$ lies on the central face of $H$, the labels of 
$vw$ are equal, i.e., $\ell_{C_1}(vw)=\ell_{C_2}(vw)$.
\end{proposition}

\begin{proposition}[Lemma~16 in 
\cite{bnrw-ttsmford-17-arxiv}]\label{lem:rect:two_cycles_horizontal}
  Let $C$ and $C'$ be two essential cycles that have at least one common 
  vertex. If all edges on $C$ are labeled with $0$, $C'$ is neither increasing 
  nor decreasing.
\end{proposition}

Proposition~\ref{lem:repr:equal_labels_at_intersection} is a useful
tool for comparing the labels of two interwoven essential cycles. For
example, if $C_1$ is decreasing, we can conclude for all edges of
$C_2$ that also lie on $C_1$ and that are incident to the central face
of $H$ that they have non-negative labels.
Proposition~\ref{lem:rect:two_cycles_horizontal} is useful in the
scenario where we have an essential cycle $C$ with non-negative
labels, and a decreasing cycle $C'$ that shares a vertex with $C$.  We can then 
conclude that $C$ is also decreasing.
In particular, these two propositions together imply that the central face of 
the graph $H$ formed by two decreasing cycles is bounded by a decreasing cycle.

\section{Symmetries and Normalization}
\label{sec:symmetries-normalization-appendix}

In our arguments we frequently exploit certain symmetries.  For
an ortho-radial representation $\Gamma$ we introduce two new
ortho-radial representations, its \emph{flip} $\flip{\Gamma}$ and its
\emph{mirror} $\mirror{\Gamma}$.  Geometrically, viewed as a drawing on a
cylinder, a flip corresponds to rotating the cylinder by $180\degree$
around a line perpendicular to the axis of the cylinder so that is
upside down, see Figure~\ref{fig:flipping-cylinder}, whereas mirroring
corresponds to mirroring it at a plane that is parallel to the axis of
the cylinder; see Figure~\ref{fig:mirroring-cylinder}.  Intuitively,
the first transformation exchanges left/right and top/bottom,
and thus preserves monotonicity of cycles, while the second
transformation exchanges left/right but not top/bottom, and
thus maps increasing cycles to decreasing ones and vice versa.
This intuition is indeed true with the correct definitions of
$\flip{\Gamma}$ and $\mirror{\Gamma}$, but due to the non-locality of
the validity condition for ortho-radial
representations and the dependence on a reference edge this requires
some care. The following two lemmas formalize flipped and mirrored
orthogonal representations. We denote the reverse of an edge $e$ by
$\reverse{e}$.

  \begin{figure}[t]
    \begin{subfigure}[c]{0.35\textwidth}
      \centering
      \includegraphics[width=\textwidth]{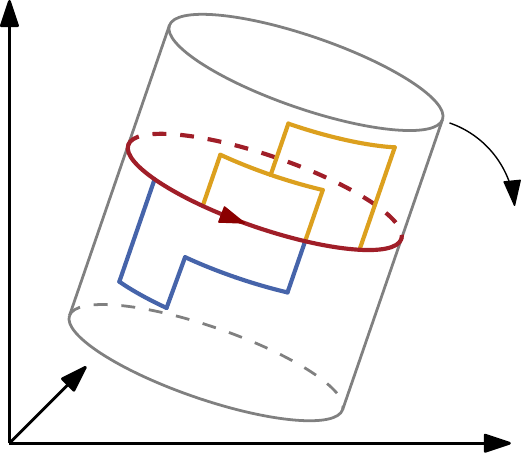}
\subcaption{Flipping the cylinder.}
    \end{subfigure} \hspace{5ex}
    \begin{subfigure}[c]{0.55\textwidth}
      \centering
      \includegraphics[page=2,width=\textwidth]{figures/flipping_cylinder.pdf}
\subcaption{The ortho-radial drawing before and after flipping.}
    \end{subfigure}
    \caption{Illustration of flipping the cylinder.}
   \label{fig:flipping-cylinder}
 \end{figure}
 
 \begin{figure}[t]
  \centering
  \includegraphics[]{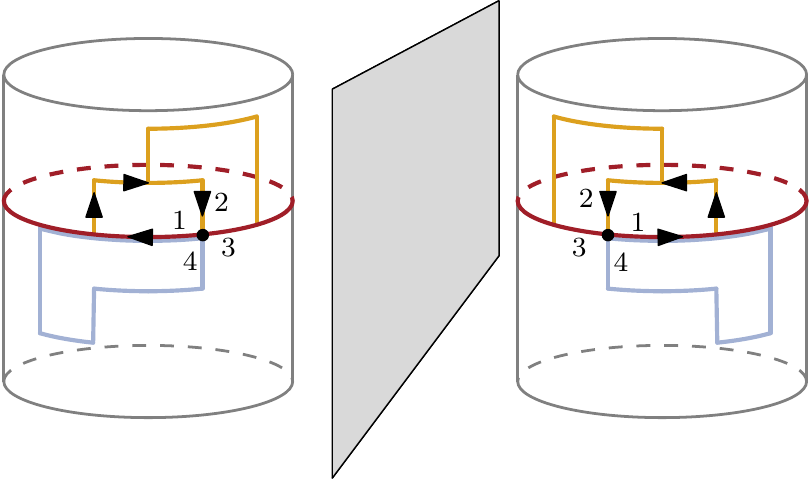}
  \caption{Mirroring the cylinder.}
  \label{fig:mirroring-cylinder}
 \end{figure}

\begin{restatable}[Flipping]{restatable-lemma}{fliplabel}
  \label{lem:flip_label}
  Let $\Gamma$ be an ortho-radial
  representation with outer face $f_o$ and central face~$f_c$. If the cycle 
  bounding the central face is not monotone, there exists an ortho-radial
  representation $\flip{\Gamma}$ such that
  \begin{compactenum}
  \item $\reverse{f}_c$ is the outer face of $\flip{\Gamma}$ and 
  $\reverse{f}_o$ is the
    central face of $\flip{\Gamma}$,
  \item $\reverse{\ell}_{\reverse{C}}(\reverse{e})=\ell_{C}(e)$ for all
  essential cycles $C$ and edges $e$ on $C$, where $\reverse{\ell}$
  is the labeling in $\reverse{\Gamma}$.
  \end{compactenum}
  In particular, increasing and decreasing cycles of $\Gamma$ correspond to
  increasing and decreasing cycles of $\flip{\Gamma}$, respectively. 
\end{restatable}

\begin{proof}
  We define $\flip{\Gamma}$ as follows.  The central face of $\Gamma$
  becomes the outer face of $\flip{\Gamma}$ and the outer face of
  $\Gamma$ becomes the central face of $\flip{\Gamma}$.  Further, we
  choose an arbitrary edge $e^{\star\star}$ on the central face~$f_c$ of
  $\Gamma$ with $\ell_{f_c}(e^{\star\star}) = 0$ (such an edge exists
  since the cycle bounding the central face is not monotone), and
  choose $\reverse{e^{\star\star}}$ as the reference edge of
  $\flip{\Gamma}$.  All other information of $\Gamma$ is transferred
  to $\flip{\Gamma}$ without modification. As the local structure is unchanged, 
  $\flip{\Gamma}$ is an ortho-radial representation.
  
  The essential cycles in $\Gamma$ bijectively correspond to the essential 
  cycles in $\flip{\Gamma}$ by reversing the direction of the cycles. That 
  is, any essential cycle $C$ in $\Gamma$ corresponds to the cycle 
  $\reverse{C}$ 
  in $\flip{\Gamma}$. Note that the reversal is necessary since we always 
  consider essential cycles to be directed such that the center lies in its 
  interior, which is defined as the area locally to the right of the cycle.
  
  Consider any essential cycle $C$ in $\Gamma$. We denote the labeling
  of $C$ in $\Gamma$ by $\ell_C$ and the labeling of $\reverse{C}$ in
  $\flip{\Gamma}$ by $\reverse{\ell}_{\reverse{C}}$. We show that for
  any edge $e$ on $C$ it is
  $\ell_C(e) = \reverse{\ell}_{\reverse{C}}(\reverse{e})$, which
  particularly implies that any monotone cycle in $\Gamma$ corresponds
  to a monotone cycle in $\reverse{\Gamma}$ and vice versa. First, we pick a 
  simple
  path $P$ from $e^\star$ to $e$ such that $P$ lies in the exterior of
  $C$ in $\Gamma$ and another simple path $Q$ from $e$ to
  $e^{\star\star}$ that lies in the interior of $C$.
    
  Assume for now that $P+e+Q$ is simple. We shall see at the end how the proof 
  can be extended if this is not the case. By the choice of $e^{\star\star}$, 
  we have
  \begin{align}
  0 &= \ell_{f_c}(e^{\star\star}) = \rot(e^\star + P + e + Q + e^{\star\star}) 
  = \rot(e^\star + P + e) + \rot(e + Q + e^{\star\star})
  \end{align}
  Hence, 
  $ \rot(e^\star + P + e) = \rot(\reverse{e^{\star\star}} +
  \reverse{Q} + \reverse{e})$ and in total
   \begin{align}
  \reverse{\ell}_{\reverse{C}}(\reverse{e}) 
  = \rot(\reverse{e^{\star\star}} + \reverse{Q} + 
  \reverse{e})
  = \rot(e^\star + P + e) 
  = \ell_C(e).
  \end{align}
  Thus, any monotone cycle in $\Gamma$ corresponds to a monotone cycle in 
  $\flip{\Gamma}$ and vice versa.
  
  \begin{figure}[t]
    \centering
    \begin{subfigure}[c]{0.3\textwidth}
      \centering
      \includegraphics[]{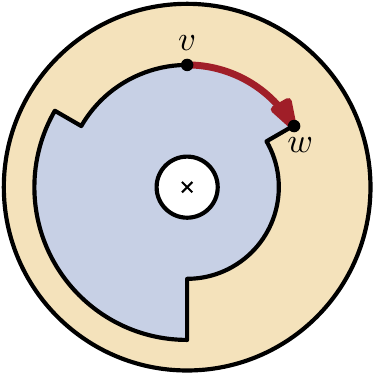}
      \subcaption{The original graph with the essential cycle~$C$.}
    \end{subfigure}
    \hspace{2ex}
    \begin{subfigure}[c]{0.3\textwidth}
      \centering
      \includegraphics[]{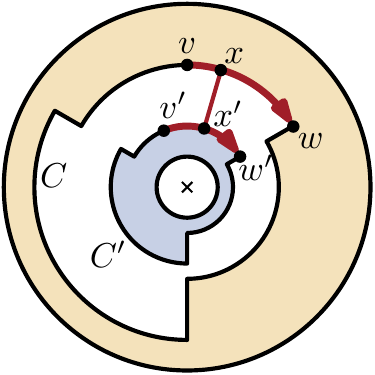}
      \subcaption{The graph after it was cut at $C$.}
    \end{subfigure}
    \caption{Cutting the graph at an essential cycle $C$.}
    \label{fig:cutting}
  \end{figure}
  
  If $P+e+Q$ is not simple, we make it simple by cutting $G$ at $C$
  such that the interior and the exterior of $C$ get their own copies
  of $C$; see Figure~\ref{fig:cutting}. We connect the two parts by an
  edge between two new vertices $x$ and $x'$ on the two copies of $e$,
  which we denote by $vw$ in the exterior part and $v'w'$ in the
  interior part.  The new edge is placed perpendicular to these
  copies.  The path $P+vxx'w'+Q$ is simple and its rotation is
  $0$. Hence, the argument above implies
  $\reverse{\ell}_{\reverse{C}}(\reverse{e}) = \ell_C(e)$.
\end{proof}

\begin{restatable}[Mirroring]{restatable-lemma}{mirrorlabel}
  \label{lem:mirroring_label}
  Let $\Gamma$ be an ortho-radial
  representation with outer face $f_o$ and central face $f_c$. There
  exists an ortho-radial representation $\mirror{\Gamma}$ such that
  \begin{compactenum}
  \item $\reverse{f}_o$ is the outer face of $\mirror{\Gamma}$ and 
  $\reverse{f}_c$ is the
    central face of $\mirror{\Gamma}$,
  \item $\mirror{\ell}_{\reverse{C}}(\reverse{e})=-\ell_{C}(e)$ for all
  essential cycles $C$ and edges $e$ on $C$, where $\mirror{\ell}$
  is the labeling in $\mirror{\Gamma}$.
  \end{compactenum}
  
  In particular, increasing and decreasing cycles of $\Gamma$ correspond to
  decreasing and increasing cycles of $\mirror{\Gamma}$, respectively. 
\end{restatable}

\begin{proof}
  \begin{figure}
    \begin{subfigure}[b]{0.45\textwidth}
      \centering
      \includegraphics[]{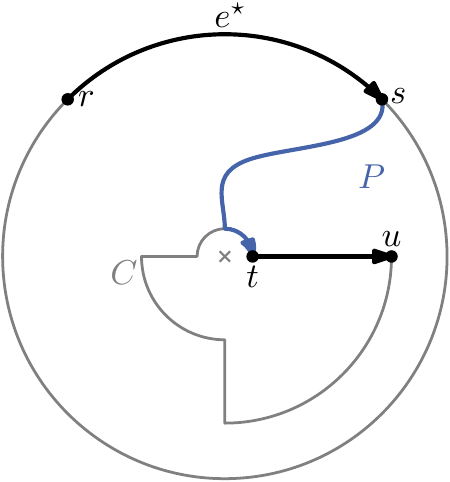}
      \subcaption{The original graph with a path $P$ from the reference edge 
        $e^\star$ to the edge $e$ on the essential cycle~$C$.}
    \end{subfigure}
    \hfill
    \begin{subfigure}[b]{0.45\textwidth}
      \centering
      \includegraphics[]{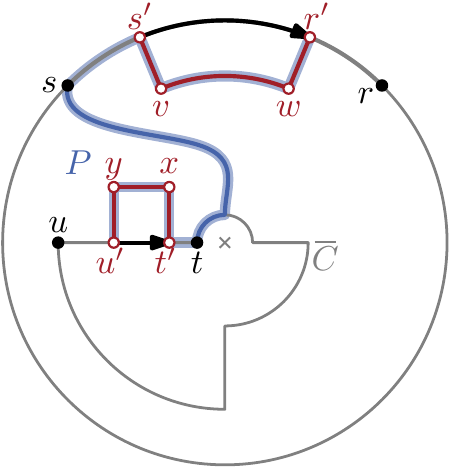}
      \subcaption{The mirrored graph including the construction to extend the 
        path $P$ to a path from $r'$ to~$u'$.}
    \end{subfigure}
    \caption{The construction that is used to adapt the path from the reference 
      edge to an edge $e$ on an essential cycle when the graph is mirrored.}
    \label{fig:mirroring_construction}
  \end{figure}

  We define $\mirror{\Gamma}$ as follows.  We reverse the direction of
  all faces and reverse the order of the edges around each vertex.
  The outer and central face are equal to those in $\Gamma$ (except
  for the directions) and the reference edge is $\reverse{e^\star}$.
  Since the reference edge is reversed, edges that point left in
  $\Gamma$ point right in $\mirror{\Gamma}$ and vice versa, but the
  edges that point up (down) in $\Gamma$ also point up (down) in
  $\mirror{\Gamma}$. Note that this construction satisfies the conditions for 
  ortho-radial representations.

  Let $e=tu$ be an edge on $C$ and $P$ a simple path from
  $e^\star=rs$ to $e=tu$ that lies in the exterior of $C$; we consider
  $P$ without $e^\star$ and $e$.  In the mirrored representation
  $\mirror{\Gamma}$ the direction of the reference edge and $C$ are
  reversed, i.e., the reference edge is $\reverse{e^\star}$, and we
  consider $\reverse{C}$.  As $P$ starts at the tail of
  $\reverse{e^\star}$ and ends at the head of $\reverse{e}$, we cannot
  simply use $P$ to compute the label
  $\mirror{\ell}_{\reverse{C}}(\reverse{e})$.
  
  Therefore, we modify $G$ and $\Gamma$ slightly by adding some
  vertices and edges as follows; see Figure~\ref{fig:mirroring_construction}. 
  The edge
  $e^\star$ is subdivided by two new vertices $r'$ and $s'$, which are
  connected by a new path $s'vwr'$ in the interior of the graph. We
  place this path such that the face formed by this path and $r's'$ is
  a rectangle.  Similarly, we add the vertices $t'$ and $u'$ on $e$
  and we connect $t'$ to $u'$ by a new path $t'xyu'$ such that the new
  face with these four vertices is a rectangle that lies in the
  exterior of $C$.  We set $r's'$ as the new reference edge and call the 
  resulting representation $\Gamma'$. Since
  $r's'$ is a part of the original reference edge $e^\star$, this
  preserves the labelings. In particular, the labels of $tt'$, $t'u'$
  and $u'u$ in $\Gamma'$ are equal to $\ell_C(e)$ in $\Gamma$.
  
  Setting $P'=r'wvs's+P+tt'xyu'$, we obtain a simple path in
  the mirrored representation of $\Gamma'$ from 
  the reference edge
  $\reverse{e^\star}$ to $\reverse{e}$. Since mirroring flips the sign of the 
  rotation of a path, we get
  \begin{align*}
  \mirror{\ell}_{\reverse{C}}(\reverse{e}) &= \rot_{\mirror{\Gamma'}}(s'r' + P' 
  + u't') 
  = \rot_{\mirror{\Gamma'}}(s'r'vws's) + \rot_{\mirror{\Gamma'}}(s's+P+tt') 
  +  \rot_{\mirror{\Gamma'}}(tt'xyu't') \\
  &= 2 - \rot_{\Gamma'}(s's+P+tt') - 2 
  = -\rot_{\Gamma}(e^\star+P+e) 
    = -\ell_C(e).    
  \end{align*}
  In particular, increasing and decreasing cycles of $\Gamma$
  correspond to decreasing and increasing cycles of $\mirror{\Gamma}$,
  respectively.
\end{proof}

We can further restrict ourselves to instances with minimum degree 2 by 
removing all degree-1 vertices in the following fashion.
Suppose $v$ is a degree-1 vertex in $G$. It is not hard to see that $G$ 
contains a
monotone cycle if and only if $G-v$ does.  It is thus tempting to
iteratively remove degree-1 vertices.  However later, when we augment
the graph and its ortho-radial representation so that all faces become
rectangular, reinserting these vertices may require non-trivial
modifications.  To avoid this we present a transformation that
produces a supergraph of a subdivision of $G$ where every vertex has
degree~2.

\begin{wrapfigure}{l}{5.1cm}
  \centering
  \includegraphics[page=2]{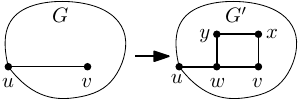}
  \caption{The degree-1 vertex $v$ is replaced by a rectangle.}
  \label{fig:normalization}
\end{wrapfigure}
Let $uv$ be the edge incident to the degree-1 vertex $v$.  We obtain a
graph $G'$ from $G$ by subdividing $uv$ with a new vertex $w$, and
adding two new vertices $x$ and $y$ along with edges $vx$, $xy$ and
$wy$.  Further, we obtain $\Gamma'$ from $\Gamma$ by setting all
angles in the inner face bounded by $v,w,x,y$ to $90\degree$; see
Fig.~\ref{fig:normalization}.  It is easy to see that $\Gamma'$ is
valid if and only if $\Gamma$ is, since a monotone cycle of $G$ is
also a monotone cycle of $G'$, and conversely, since $uw$ is a bridge,
a cycle $C$ of $G'$ is either $w,v,x,y$ or it is contained in $G-uv$.
The former cycle is non-essential by construction, and hence any
monotone cycle of $G'$ is also a monotone cycle of $G$.  It follows
that the monotone cycles of $G$ bijectively correspond to the monotone
cycles of $G'$.  Iteratively applying this construction to all degree-1 
vertices yields the
following lemma.

\begin{lemma}
  \label{lem:degree-1-removal}
  Let $G$ be a planar 4-graph on $n$ vertices with 
  ortho-radial
  representation $\Gamma$.  In $\O(n)$ time we can compute a supergraph
  $G^{*}$ with minimum degree~2 of a subdivision of $G$ with
  ortho-radial representation $\Gamma^{*}$ such that there is a
  bijective correspondence between monotone cycles in $G$ and monotone
  cycles in $G^{*}$.
\end{lemma}

\section{Finding Monotone Cycles}\label{sec:finding_monotone_cycles}

The two conditions for ortho-radial representations are local and checking them 
can easily be done in linear time. We therefore assume in this section that we 
are given a planar 4-graph $G$ with an ortho-radial representation $\Gamma$.
The condition for validity however references all essential cycles of which 
there may be exponentially many. We
present an algorithm that checks whether $\Gamma$ contains a
monotone cycle and computes such a cycle if one exists.  The main
difficulty is that the labels on a decreasing cycle $C$
depend on an elementary path $P$ from the reference edge to $C$.
However, we know neither the path $P$ nor the cycle $C$ in advance,
and choosing a specific cycle $C$ may rule out certain paths $P$ and
vice versa.

We only describe how to search for decreasing cycles; increasing
cycles can be found by searching for decreasing cycles in the mirrored
representation by Lemma~\ref{lem:mirroring_label}.  A decreasing cycle
$C$ is \emph{outermost} if it is not contained in the interior of any
other decreasing cycle.  Clearly, if $\Gamma$ contains a decreasing
cycle, then it also has an outermost one.  We first show that in this
case this cycle is uniquely determined.

\begin{restatable}{restatable-lemma}{labeloutercycleunique}
  \label{lem:outermost_decreasing_cycle}
  If $\Gamma$ contains a decreasing cycle, there is a unique
  outermost decreasing cycle.
\end{restatable}

\begin{proof}
  Assume that $\Gamma$ has two outermost decreasing cycles $C_1$ and
  $C_2$, i.e., $C_1$ does not lie in the interior of $C_2$ and vice
  versa.  Let $C$ be the cycle bounding the outer face of the subgraph
  $H=C_1+C_2$ that is formed by the two decreasing cycles. By
  construction, $C_1$ and $C_2$ lie in the interior of $C$, and we
  claim that $C$ is a decreasing cycle contradicting that $C_1$ and
  $C_2$ are outermost. To that end, we show that
  $\ell_C(e)=\ell_{C_1}(e)$ for any edge~$e$ that belongs to both $C$
  and $C_1$, and $\ell_C(e)=\ell_{C_2}(e)$ for any edge $e$ that
  belongs to both $C$ and $C_2$. Hence, all edges of $C$ have a
  non-negative label since $C_1$ and~$C_2$ are decreasing. By
  Proposition~\ref{lem:rect:two_cycles_horizontal} there is at least
  one label of $C$ that is positive, and hence $C$ is a decreasing
  cycle.

  It remains to show that $\ell_C(e)=\ell_{C_1}(e)$ for any edge~$e$
  that belongs to both $C$ and $C_1$; the case that $e$ belongs to
  both $C$ and $C_2$ can be handled analogously.  Let $\Gamma_H$ be
  the ortho-radial representation~$\Gamma$ restricted to $H$.  We flip
  the cylinder to exchange the outer face with the central face and
  vice versa. More precisely, Lemma~\ref{lem:flip_label} implies that
  the reverse edge $\reverse{e}$ of $e$ lies on the central face of
  the flipped representation $\flip{\Gamma}_H$ of $\Gamma_H$. Further,
  it proves that
  $\reverse{\ell}_{\reverse{C}}(\reverse{e})=\ell_{C}(e)$ and
  $\reverse{\ell}_{\reverse{C_1}}(\reverse{e})=\ell_{C_1}(e)$, where
  $\reverse{\ell}$ is the labeling in $\reverse{\Gamma}_H$.  Hence, by
  Proposition~\ref{lem:repr:equal_labels_at_intersection} we obtain
  $\reverse{\ell}_{\reverse{C}}(\reverse{e})=
  \reverse{\ell}_{\reverse{C_1}}(\reverse{e})$.  Flipping back the cylinder, again by
  Lemma~\ref{lem:flip_label} we obtain $\ell_C(e)=\ell_{C_1}(e)$. %
\end{proof}

The core of our algorithm is an adapted left-first DFS. Given a
directed edge~$e$ it determines the outermost decreasing cycle $C$ in
$\Gamma$ such that $C$ contains $e$ in the given direction and $e$ has
the smallest label among all edges on $C$, if such a cycle exists.  By
running this test for each directed edge of $G$ as the start edge, we
find a decreasing cycle if one exists.

\newcommand{\reference}{\mathrm{ref}}

Our algorithm is based on a DFS that visits each vertex at most once.
A left-first search maintains for each visited vertex $v$ a reference
edge $\reference(v)$, the edge of the search tree via which $v$ was
visited. Whenever it has a choice which vertex to visit next, it
picks the first outgoing edge in clockwise direction after the
reference edge that leads to an unvisited vertex.  In addition to
that, we employ a filter that ignores certain outgoing edges during
the search.  To that end, we define for all outgoing edges $e$
incident to a visited vertex $v$ a \emph{search label}~$\tilde\ell(e)$
by setting
$\tilde\ell(e) = \tilde\ell(\reference(v)) + \rot(\reference(v) + e)$
for each outgoing edge $e$ of $v$.  In our search we ignore edges with
negative search labels.  For a given directed edge $vw$ in $G$ we 
initialize the search by
setting $\reference(w) = vw$, $\tilde\ell(vw) = 0$ and then start
searching from $w$.

Let $T$ denote the directed search tree with root $w$ constructed by
the DFS in this fashion.  If $T$ contains $v$, then this determines a
\emph{candidate cycle} $C$ containing the edge $vw$. If $C$ is a
decreasing cycle, which we can easily check by determining an
elementary path from the reference edge to $C$, we report it.
Otherwise, we show that there is no outermost decreasing cycle $C$
such that $vw$ lies on $C$ and has the smallest label among
all edges on $C$.

It is necessary to check that $C$ is essential and
decreasing.  For example the cycle in
Figure~\ref{fig:nondecreasing_cycle_found} is found by the search and
though it is essential, it is non-decreasing.  This is caused by the
fact that the label of $vw$ is actually $-4$ on this cycle but the
search assumes it to be $0$.

\begin{figure}
  \centering
  \begin{minipage}[b]{.45\textwidth}
    \centering
    \includegraphics[]{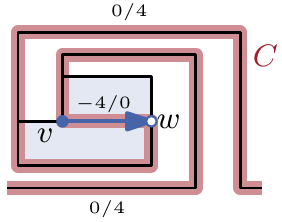}\vspace{-1ex}    
  \end{minipage}\hfill
  \begin{minipage}[b]{.45\textwidth}
    \centering
    \includegraphics{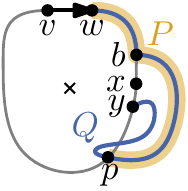}\vspace{-1ex}    
  \end{minipage}
  
  \begin{minipage}[t]{.49\textwidth}
    \caption{The search from $vw$ finds the non-decreasing cycle $C$.  Edges are labeled $\ell_{C}(e) / \tilde\ell(e)$.}
    \label{fig:nondecreasing_cycle_found}
  \end{minipage}\hfill
  \begin{minipage}[t]{.45\textwidth}
    \caption{Path $Q$ and its prefix $P$ that leaves $C$ once and ends 
      at a vertex $p$ of $C$.}
    \label{fig:prefix}
  \end{minipage}
\end{figure}

\begin{lemma}\label{lem:dfs_correctness}
  Assume $\Gamma$ contains a decreasing cycle.  Let $C$ be the
  outermost decreasing cycle of $\Gamma$ and let $vw$ be an edge on
  $C$ with the minimum label, i.e., $\ell_C(vw) \leq \ell_C(e)$ for
  all edges~$e$ of~$C$.  Then the left-first DFS from $vw$ finds $C$.
\end{lemma}
\begin{proof}
  Assume %
  that the search does not find
  $C$.  Let $T$ be the tree formed by the edges visited by the search.
  Since the search does not find $C$ by assumption, a part of
  $\subpath{C}{w,v}$ does not belong to $T$.  Let $xy$ be the first
  edge on $\subpath{C}{w,v}$ that is not visited, i.e.,
  $\subpath{C}{w,x}$ is a part of $T$ but $xy\not\in T$.  There are
  two possible reasons for this.  Either $\tilde\ell(xy) < 0$ or $y$
  has already been visited before via another path $Q$ from $w$ with
  $Q\neq \subpath{C}{w,y}$.  The case $\tilde{\ell}(xy)<0$ can be
  excluded as follows.  By the construction of the labels
  $\tilde\ell$, for any path $P$ from $w$ to a vertex $z$ in $T$ and any edge
  $e'$ incident to $z$ we have $\tilde\ell(e') = \rot(vw+P+e')$.  In
  particular, $\tilde{\ell}(xy) = \rot(\subpath{C}{vw, xy}) = \ell_C(xy) -
  \ell_C(vw) \geq 0$ since the rotation can be rewritten as a label difference 
  (see~\cite[Obs.~7]{bnrw-ttsmford-17-arxiv}) and $vw$ has the smallest label 
  on $C$.
  
  Hence, $T$ contains a path $Q$ from $w$ to $x$ that was found by the
  search before and $Q$ does not completely lie on $C$. There is a
  prefix of $Q$ (possibly of length $0$) lying on $C$ followed by a
  subpath not on $C$ until the first vertex $p$ of $Q$ that again
  belongs to $C$; see Figure~\ref{fig:prefix}.  We set $P=\subpath{Q}{w, p}$  
  and denote the vertex where $P$ leaves $C$
  by $b$.  By construction the edge $vw$ lies on $\subpath{C}{p,b}$.
  The subgraph $H=P+C$ that is formed by the decreasing cycle $C$ and
  the path $P$ consists of the three internally
  vertex-disjoint paths $\subpath{P}{b,p}$, $\subpath{C}{b,p}$ and
  $\subpath{\reverse{C}}{b,p}$ between $b$ and $p$.  Since edges that
  are further left are preferred during the search, the clockwise
  order of these paths around $b$ and $p$ is fixed.  In $H$ there are
  three faces, bounded by $C$,
  $\subpath{\reverse{C}}{b,p}+\subpath{\reverse{P}}{p,b}$ and
  $\subpath{P}{b,p}+\subpath{\reverse{C}}{p,b}$, respectively.  Since
  $C$ is an essential cycle and a face in $H$, it is the central face
  and one of the two other faces is the outer face. These two
  possibilities are shown in Figure~\ref{fig:search_path}.  We denote
  the cycle bounding the outer face but in which the edges are
  directed such that the outer face lies locally to the left by
  $C'$. That is, the boundary of the outer face is $\reverse{C'}$.  We
  distinguish cases based on which of the two possible cycles
  constitutes $\reverse{C'}$.

    \begin{figure}
    \centering
    \begin{subfigure}[b]{0.43\textwidth}
      \centering
      \includegraphics[]{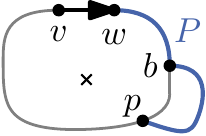}
      \subcaption{The edge $vw$ lies on the outer face of~$H$.}
      \label{fig:search_path-forward}
    \end{subfigure}
    \hspace{1em}
    \begin{subfigure}[b]{0.52\textwidth}
      \centering
      \includegraphics[]{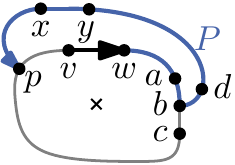}
      \subcaption{The edge $vw$ does not lie on the outer face of $H$.}
      \label{fig:search_path-backward}
    \end{subfigure} 
    \caption{The two possible embeddings of the subgraph formed by the  
    decreasing cycle $C$ and the path $P$, which was found by the  
    search.}
    \label{fig:search_path}
  \end{figure}
  
  If
  $\reverse{C'}=\subpath{\reverse{C}}{b,p}+\subpath{\reverse{P}}{p,b}$
  forms the outer face of $H$, $vw$ lies on $C'$ as illustrated in
  Figure~\ref{fig:search_path-forward} and we show that $C'$ is a
  decreasing cycle, which contradicts the assumption that $C$ is the
  outermost decreasing cycle.  Since $P$ is simple and lies in the
  exterior of $C$, the path $P$ is contained in $C'$, which
  means $\subpath{C'}{w, p}= P$. The other part of $C'$ is formed by
  $\subpath{C}{p, w}$. Since $C$ forms the central face of $H$, the
  labels of the edges on $\subpath{C}{p,w}$ are the same for $C$ and
  $C'$ by Proposition~\ref{lem:repr:equal_labels_at_intersection}. In
  particular, $\ell_C(vw) = \ell_{C'}(vw)$ and all the labels of edges
  on $\subpath{C}{p,w}$ are non-negative because $C$ is
  decreasing. The label of any edge $e$ on both $C'$ and $P$ is
  $\ell_{C'}(e) = \ell_{C'}(vw) + \rot(vw + \subpath{P}{w, e}) =
  \ell_C(vw) + \tilde{\ell}(e) \geq 0$.  Thus, the labeling of $C'$ is
  non-negative. Further, not all labels of $C'$ are $0$ since
  otherwise $C$ would not be a decreasing cycle by
  Proposition~\ref{lem:rect:two_cycles_horizontal}.  Hence, $C'$
  is decreasing and contains $C$ in its interior, a contradiction.
  
  If $\reverse{C'} = \subpath{\reverse{C}}{p,b}+\subpath{P}{b,p}$, 
  the edge $vw$ does not lie on $C'$; see 
  Figure~\ref{fig:search_path-backward}.
  We show that $C'$ is a decreasing cycle containing $C$ in 
  its interior, again contradicting the choice of $C$.
  As above, Proposition~\ref{lem:repr:equal_labels_at_intersection}
  implies that the common edges of $C$ and $C'$ have the same labels
  on both cycles.  It remains to show that all edges $xy$ on
  $\subpath{\reverse{P}}{p,b}$ have non-negative labels.  To establish
  this we use paths to the edge that follows $b$ on $C$.
  This edge $bc$ has the same label on both cycles and thus provides a handle
  on $\ell_{C'}(xy)$.  We make use of the following equations, which
  follow immediately from the definition of the (search) labels.
  \begin{align*}
    \ell_{C'}(bc) &= \ell_{C'}(xy) + \rot(\subpath{\reverse{P}}{xy,db}) + \rot(dbc) = \ell_{C'}(xy) - \rot(\subpath{P}{bd,yx}) - \rot(cbd)\\
    \ell_{C}(bc)  &= \ell_{C}(vw)  + \rot(\subpath{C}{vw,ab})           + \rot(abc)\\
    \tilde\ell(yx) &= \rot(\subpath{C}{vw,ab}) + \rot(abd) + \rot(\subpath{P}{bd,yx})
  \end{align*}
  Since $\ell_C(bc) = \ell_{C'}(bc)$ and $\rot(abd) = -\rot(dba)$, we
  thus get
  \begin{align*}
    \ell_{C'}(xy) &= \ell_{C}(vw) + \rot(\subpath{C}{vw,ab}) + \rot(abc) +\rot(\subpath{P}{bd,yx}) + \rot(cbd)\\
                 &= \ell_{C}(vw) + \tilde\ell(yx) + \rot(dba) + \rot(abc) + \rot(cbd).
  \end{align*}
  Since $\ell_C(vw) \ge 0$ and $\tilde\ell(yx) \ge 0$ (as $yx$ was not filtered 
  out), it follows that
  $\ell_{C'}(xy) \ge \rot(abc) + \rot(dba) + \rot(cbd) = 2$ as this is
  the sum of clockwise rotations around a degree-3 vertex.  Hence,
  $C'$ is decreasing and contains $C$ in its interior, a
  contradiction.
  Since both embeddings of $H$ lead to a contradiction, we obtain a
  contradiction to our initial assumption that the search fails to
  find $C$.
\end{proof}

  The left-first DFS clearly runs in $\O(n)$ time.  In order to
  guarantee that the search finds a decreasing cycle if one exists, we
  run it for each of the $O(n)$ directed edges of $G$.  Since some
  edge must have the lowest label on the outermost decreasing cycle,
  Lemma~\ref{lem:dfs_correctness} guarantees that we eventually find a
  decreasing cycle if one exists.  Increasing cycles can be found by
  finding decreasing cycles in the mirror representation
  $\mirror{\Gamma}$ (Lemma~\ref{lem:mirroring_label}).

\begin{theorem}
  \label{thm:test-valid}
  Let $G$ be a planar 4-graph on $n$ vertices and let 
  $\Gamma$ be an
  ortho-radial representation of $G$.  It can be determined in
  $\O(n^2)$ time whether $\Gamma$ is valid.
\end{theorem}

\section{Rectangulation}\label{sec:rectangulation}

\begin{figure}[t]
  \centering
     \begin{minipage}[b]{0.48\textwidth}
    \begin{subfigure}[b]{\textwidth}
      \centering
      \includegraphics[]{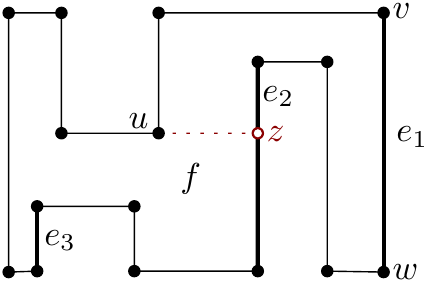}
      \subcaption{Multiple edge candidates.}
      \label{fig:hshape}
    \end{subfigure}    
  \end{minipage}
    \begin{minipage}[b]{0.25\textwidth}
    \begin{subfigure}[b]{\textwidth}
      \centering
      \includegraphics[page=2,scale=1]{figures/hshape.pdf}
      \subcaption{Vertical. }
      \label{fig:aug-vert}
    \end{subfigure}
    \begin{subfigure}[b]{\textwidth}
      \centering
      \includegraphics[page=3]{figures/hshape.pdf}
      \subcaption{Horizontal. }
      \label{fig:aug-horz-decreasing}
    \end{subfigure}
    \end{minipage}
    \begin{minipage}[b]{0.25\textwidth}
    \begin{subfigure}[b]{\textwidth}
      \centering
      \includegraphics[page=4]{figures/hshape.pdf}
      \subcaption{Horizontal. }
      \label{fig:aug-horz-valid}
    \end{subfigure}
    \begin{subfigure}[b]{\textwidth}
      \centering
      \includegraphics[page=5]{figures/hshape.pdf}
      \subcaption{Horizontal. }
      \label{fig:aug-path}
    \end{subfigure}
   \end{minipage}
  
  \caption{Examples of augmentations.
    (\protect\subref{fig:hshape})~The candidate edges of $u$ are $e_1$, $e_2$ 
    and $e_3$.
    (\protect\subref{fig:aug-vert})~Insertion of vertical edge $uz$. 
    (\protect\subref{fig:aug-horz-decreasing})~$\Gamma^{u}_{vw}$ contains a 
    decreasing cycle.
    (\protect\subref{fig:aug-horz-valid})~$\Gamma^{u}_{vw}$ is
    valid.
    (\protect\subref{fig:aug-path})~Insertion of horizontal edge 
    $uw_{i}$ because there is a horizontal path from $w_i$ to $u$.}
    \label{fig:augmentation}
\end{figure}

The core of the algorithm for drawing a valid ortho-radial
representation $\Gamma$ of a graph $G$ by Barth et
al.~\cite{bnrw-ttsmford-17} is a \emph{rectangulation procedure} that
successively augments $G$ with new vertices and edges to a graph
$G^{*}$ along with a valid ortho-radial representation $\Gamma^{*}$
where every face of $G^{*}$ is a \emph{rectangle}.  %
A regular face is a rectangle if it has exactly four turns, which are all
right turns. The outer and central faces are rectangles if they have no turns.
The ortho-radial representation~$\Gamma^{*}$ is then drawn by computing 
flows in two flow networks~\cite[Thm.\ 
18]{bnrw-ttsmford-17-arxiv}.

To facilitate the analysis, we
briefly sketch the augmentation procedure.  Here it is crucial that we
assume our instances to be normalized; in particular they do not have
degree-1 vertices.
The augmentation algorithm works by augmenting non-rectangular faces
one by one, thereby successively removing concave angles at the
vertices until all faces are rectangles.  Consider a face $f$ with a
left turn (i.e., a concave angle) at $u$ such that the following two
turns when walking along $f$ (in clockwise direction) are right turns;
see Figure~\ref{fig:augmentation}. We call $u$ a \emph{port} of $f$.
We define a set of \emph{candidate edges} that contains precisely
those edges $vw$ of $f$, for which $\rot(\subpath{f}{u, vw}) = 2$; see
Figure~\ref{fig:hshape}.  We treat this set as a sequence, where the
edges appear in the same order as in $f$, beginning with the first
candidate after $u$.  The \emph{augmentation} $\Gamma^u_{vw}$ with
respect to a candidate edge $vw$ is obtained by splitting the edge
$vw$ into the edges $vz$ and $zw$, where $z$ is a new vertex, and
adding the edge $uz$ in the interior of $f$ such that the angle formed
by $zu$ and the edge following $u$ on $f$ is $90\degree$.  The
direction of the new edge $uz$ in $\Gamma^{u}_{vw}$ is the same for
all candidate edges.  If this direction is vertical, we call $u$ a
\emph{vertical port} and otherwise a \emph{horizontal port}. We note
that any vertex with a concave angle in a face becomes a port during
the augmentation process. In particular, the incoming edge of the
vertex determines whether the port is horizontal or
vertical.  The condition for candidates guarantees that
$\Gamma^{u}_{vw}$ is an ortho-radial representation. It may, however,
not be valid.  The crucial steps in~\cite{bnrw-ttsmford-17} are
establishing the following facts.
\begin{compactenum}[{Fact} 1)]
\item Let $u$ be a vertical port. Augmenting with the first candidate
  never produces a monotone cycle~\cite[Lemma
  21]{bnrw-ttsmford-17-arxiv}.\label{fact:vertical-candidate}
\item Let $u$ be a horizontal port. Augmenting with the first
  candidate never produces an increasing cycle~\cite[Lemma
  22]{bnrw-ttsmford-17-arxiv} and augmenting with the last candidate
  never produces a decreasing cycle~\cite[Lemma
  24]{bnrw-ttsmford-17-arxiv}.\label{fact:first-last}
\item Let $u$ be a horizontal port. If two consecutive candidates
  $e_{i}=v_iw_i$ and $e_{i+1}=v_{i+1}w_{i+1}$ produce a decreasing and
  an increasing cycle, respectively, then $w_{i}$, $v_{i+1}$ and $u$
  lie on a path that starts at $w_{i}$ or
  $v_{i+1}$, and whose edges all point right; see
  Figure~\ref{fig:aug-path}. A suitable
  augmentation can be found in $O(n)$ time.
  \cite[Lemmas 25, 26]{bnrw-ttsmford-17-arxiv}.\label{fact:consecutive-done}
\end{compactenum}
It thus suffices to test for each candidate whether $\Gamma^{u}_{vw}$
is valid until either such a \emph{valid augmentation} is found or we
find two consecutive candidate edges where the first produces a decreasing cycle
and the second produces an increasing cycle. Then,
Fact~\ref{fact:consecutive-done} yields the desired valid
augmentation.
Since each valid augmentation reduces the number of concave angles, we
obtain a rectangulation after $O(n)$ valid augmentations.  Moreover,
there are $O(n)$ candidates for each augmentation, each of which can
be tested for validity (and increasing/decreasing cycles can be
detected) in $O(n^2)$ time by Theorem~\ref{thm:test-valid}. Thus, the
augmentation algorithm can be implemented to run in $O(n^{4})$ time.

In the remainder of this section we present an improvement to $O(n^{2})$ time, 
which is achieved in two steps.  First, we show that due to
the nature of augmentations the validity test can be done in $O(n)$
time (Section~\ref{sec:faster-validity-test-appendix}).  Second, for each augmentation we execute a post-processing that
reduces the number of validity tests  to
$O(n)$ in total (Section~\ref{sec:2-phase-appendix}).

\subsection{1st Improvement -- Faster Validity Test}
\label{sec:faster-validity-test-appendix}
\label{sec:faster-validity-test-main}

The general test for monotone cycles performs one left-first depth
first search per edge and runs in $\O(n^2)$ time. However, we can
exploit the special structure of the augmentation to reduce the
running time to $\O(n)$. For the proof we restrict
ourselves to the case that the inserted edge $uz$ points to the right.
The case that it points left can be handled by flipping the
representation using Lemma~\ref{lem:flip_label}.
 
The key result is that in any decreasing cycle of an augmentation the new
edge $uz$ has the minimum label.  Thus, performing only one left-first DFS 
starting at $uz$ is sufficient.  For increasing
cycles the arguments do not hold, but in a second step we show that
the test for increasing cycles can be replaced by a simple test for
horizontal paths.

Recall that the augmentations $\Gamma^u_{vw}$ that are tested during
the rectangulation are built by adding one edge $uz$ to a valid
representation~$\Gamma$. Hence, any monotone cycle in
$\Gamma^u_{vw}$ contains the edge $uz$.
 
We first show that the new edge $uz$ has label $0$ on any
decreasing cycle in the augmentation $\Gamma^u_{vw}$ if $vw$ is the first
candidate. We extend this result afterwards to augmentations to all
candidates. Since the label of edges on decreasing cycles is
non-negative, this implies in particular that the label of $uz$ is
minimum, which is sufficient for the left-first DFS to succeed (see
Lemma~\ref{lem:dfs_correctness}).

\begin{figure}[t]
    \centering
    \includegraphics{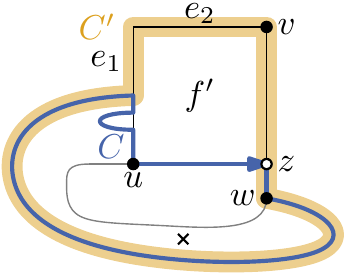}
    \caption{A decreasing cycle $C$ that uses $uz$ and an essential 
      cycle $C'$ derived from $C$.}
    \label{fig:first_candidate}
\end{figure}

\begin{lemma}
  \label{lem:label_first_candidate}
  Let $vw$ be the first candidate on $f$ after $u$. If $\Gamma^u_{vw}$ 
  contains a decreasing cycle~$C$, then $C$ contains $uz$ in this direction 
  and $\ell_C(uz) = 0$.
\end{lemma}
  
\begin{proof}
  This proof uses ideas from the proof of Lemma~22 of \cite{bnrw-ttsmford-17-arxiv}.
  We first consider the case that $C$ uses $uz$ (and not $zu$) and assume for 
  the sake of contradiction that $\ell_C(uz)\neq 0$; see 
  Figure~\ref{fig:first_candidate}. Since $uz$ points right, $\ell_C(uz)$ is 
  divisible by $4$. Together with $\ell_C(uz) \ge 0$ because $C$ is 
  decreasing, we obtain $\ell_C(uz) \geq 4$. By Lemma~14 of 
  \cite{bnrw-ttsmford-17-arxiv} there is an essential cycle $C'$ without $uz$ in the 
  subgraph $H$ that is formed by the new rectangular face $f'$ and $C$. The 
  labels of any 
  common edge $e$ of $C$ and $C'$ are equal and $\ell_{C'}(e) = \ell_C(e) \geq 
  0$. All other edges 
  of $C'$ lie on $f'$. Since $f'$ is rectangular, the labels of 
  these 
  edges differ by at most $1$ from $\ell_C(uz)$. By assumption it is 
  $\ell_C(uz) 
  \geq 4$ and therefore $\ell_{C'}(e)\geq 3$ for all edges $e\in C'\cap f'$. 
  Hence, $C'$ is a decreasing cycle in $G$ contradicting the validity of 
  $\Gamma$.
  
  If $zu\in C$, it is $\ell_C(zu)\geq 2$ and a similar argument yields a 
  decreasing cycle in~$\Gamma$.
\end{proof}

While the same statement does not generally hold for all candidates,
it does hold if the first candidate creates a decreasing cycle.

\begin{figure}[t]
  \centering
  \includegraphics{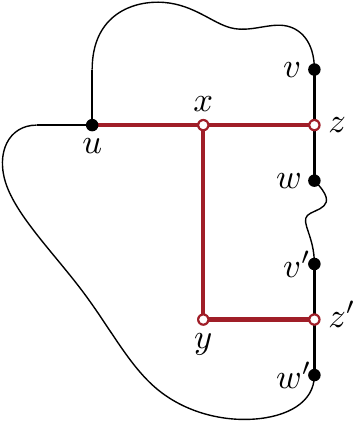}
  \caption{The structure used to simulate the simultaneous insertion of $uz$ to 
    $vw$ and $uz'$ to $v'w'$.}
  \label{fig:structure}
\end{figure}

\begin{lemma}
  \label{lem:label_any_candidate}
  Let $vw$ be the first candidate and $v'w'$ be another candidate. Denote the 
  edge inserted in $\Gamma^u_{v'w'}$ by $uz'$.
  If $\Gamma^u_{vw}$ contains a decreasing cycle, any decreasing cycle $C'$ 
  in $\Gamma^u_{v'w'}$ uses $uz'$ in this direction 
  and $\ell_{C'}(uz') = 0$.
\end{lemma}

\begin{proof}
  In order to simulate the insertion of 
  two new edges to both $vw$ and $v'w'$ we use the structure from
  the proof of Lemma~25 of~\cite{bnrw-ttsmford-17-arxiv}; see 
  Figure~\ref{fig:structure}. We denote the resulting 
  augmented representation by $\tilde\Gamma$.
  There is a one-to-one correspondence between decreasing cycles in 
  $\Gamma^u_{vw}$ and decreasing cycles in $\tilde{\Gamma}$ containing $uxz$. 
  Let $C$ be a decreasing cycle in $\tilde{\Gamma}$ containing $uxz$. By 
  Lemma~\ref{lem:label_first_candidate} the cycle $C$ contains $uxz$ in this 
  direction, and 
  we have $\ell_{C}(ux) = 0$.
  
  Similarly, for any decreasing cycle in $\Gamma^u_{v'w'}$ there is a 
  decreasing 
  cycle in $\tilde{\Gamma}$ where $uz'$ ($z'u$) is replaced by the path $uxyz'$ 
  ($z'yxu$). Let $\tilde{C'}$ be the decreasing cycle in $\tilde{\Gamma}$ that 
  corresponds to the decreasing cycle $C'$ in $\Gamma^u_{v'w'}$. We have 
  $\ell_{\tilde{C'}}(ux) = \ell_{C'}(uz')$.
  
  Suppose for now that $C'$ uses $uz'$ in this direction, which means that  
  $\tilde{C'}$ uses $ux$. Let $\tilde{f}$ be the central face of 
  $H=C+\tilde{C'}$. Either $ux$ lies on the boundary of $\tilde{f}$ or not. 
  Assume for the sake of contradiction that $ux$ does not lie on $\tilde{f}$. 
  Then, $\tilde{f}$ includes neither $xz$, $xy$ nor $yz'$. Hence, 
  $\tilde{f}$ is formed exclusively by edges that are present in $G$. Since 
  $\Gamma$ is valid, either all labels of $\tilde{f}$ are $0$ or there is an 
  edge 
  on $\tilde{f}$ with a negative label.
  
  In the first case there is an edge $e$ of $C$ leaving $\tilde{f}$, i.e., $e$ 
  starts at a vertex of $\tilde{f}$ and ends at a vertex in the exterior of 
  $\tilde{f}$.
  If no such edge existed, $\tilde{f}$ would be formed exclusively by edges of 
  $C$ or exclusively by edges of $\tilde{C'}$. This would imply that one of 
  these 
  cycles is not simple.
  But for that edge~$e$ it is $\ell_C(e) = -1$ contradicting the assumption 
  that 
  $C$ is decreasing.
  
  In the second case there is an edge~$e$ on $\tilde{f}$ with
  $\ell_{\tilde{f}}(e)<0$. This edge belongs to at least one of the
  cycles $C$ and $\tilde{C'}$, say $C$. But then it is
  $\ell_C(e) = \ell_{\tilde{f}}(e)<0$ by Proposition~\ref{lem:repr:equal_labels_at_intersection},
  contradicting again that $C$ is decreasing.
  Thus, $ux$ lies on $\tilde{f}$, and therefore we obtain from
  Proposition~\ref{lem:repr:equal_labels_at_intersection} that
  $\ell_{\tilde{C'}}(ux)=\ell_{C}(ux)=0$, where the last equality
  follows from Lemma~\ref{lem:label_first_candidate}.
  
  Above we assumed that $\tilde{C'}$ uses $ux$ in this direction. This is in 
  fact 
  the only possibility. Assume for the sake of contradiction that 
  $xu\in\tilde{C'}$. If $ux$ does not lie on the central face $\tilde{f}$ (in 
  any 
  direction), we obtain a contradiction as above. Since $C$ is essential and 
  includes $ux$, the central face lies locally to the right of $ux$. Similarly, 
  $\tilde{C'}$ is essential and contains $xu$. Therefore, the central face lies 
  to the right of $xu$ and thus to the left of $ux$. As $\tilde{f}$ cannot be 
  both to the left and the right of $ux$, we have a contradiction.
\end{proof}

\begin{figure}[t]
    \centering
    \includegraphics{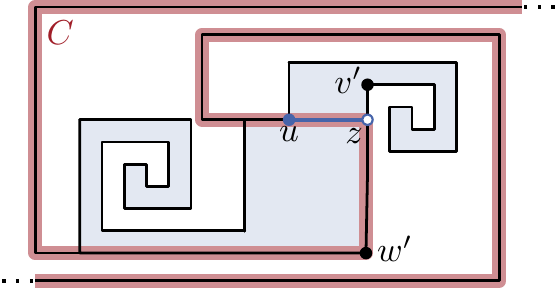}
    \caption{Here, the insertion of the edge~$uz$ to the last 
      candidate~$v'w'$ introduces an increasing cycle~$C$ with $\ell_C(uz) = 
      -4$.}
    \label{fig:increasing_cycle_negative}
\end{figure}

Altogether, we can efficiently test which of the candidates~$e_1,\dots,e_k$ 
produce decreasing cycles as follows. By Lemma~\ref{lem:label_first_candidate}, 
if the first candidate is not valid, then $\Gamma^{u}_{e_{1}}$ has
a decreasing cycle that contains the new edge $uz$ with label~$0$,
which is hence the minimum label for all edges on the cycle.  This can
be tested in $O(n)$ time by Lemma~\ref{lem:dfs_correctness}.
Fact~\ref{fact:first-last} guarantees that we either find a valid
augmentation or a decreasing cycle.  In the former case we are done,
in the second case Lemma~\ref{lem:label_any_candidate} allows us to
similarly restrict the labels of $uz$ to $0$ for the remaining
candidate edges, thus allowing us to detect decreasing cycles in
$\Gamma^u_{e_{i}}$ in $O(n)$ time for $i=2,\dots,k$.

It is tempting to use the mirror symmetry
(Lemma~\ref{lem:mirroring_label},
Appendix~\ref{sec:symmetries-normalization-appendix}) to exchange
increasing and decreasing cycles to deal with increasing cycles in an
analogous fashion.  However, this fails as mirroring invalidates the
property that $u$ is followed by two right turns in clockwise
direction. For example, in Figure~$\ref{fig:increasing_cycle_negative}$ 
inserting the edge to the last candidate introduces an increasing cycle~$C$ 
with $\ell_{C}(uz)=-4$. We therefore give a direct algorithm for
detecting increasing cycles in this case.  %

Let $e_{i}=v_iw_i$ and $e_{i+1}={v_{i+1}w_{i+1}}$ be two consecutive
candidates for $u$ such that $\Gamma^u_{e_{i}}$ contains a
decreasing cycle but $\Gamma^u_{e_{i+1}}$ does not. If
$\Gamma^u_{e_{i+1}}$ contains an increasing cycle, then by
Fact~\ref{fact:consecutive-done} the vertices $w_{i}$, $v_{i+1}$ and
$u$ lie on a path that starts at $w_{i}$ or $v_{i+1}$, and whose edges
all point right. The presence of such a horizontal path $P$ can
clearly be checked in linear time, thus allowing us to also detect
increasing cycles provided that the previous candidate produced a
decreasing cycle. If $P$ exists, we insert the edge $uw_i$ or
$uv_{i+1}$ depending on whether~$P$ starts at $w_i$ or $v_{i+1}$,
respectively; see Figure~\ref{fig:aug-path} for the first case. By
Proposition~\ref{lem:rect:two_cycles_horizontal} this does not produce
monotone cycles. Otherwise, if $P$ does not exist, the augmentation
$\Gamma^u_{e_{i+1}}$ is valid. In both cases we have resolved the
horizontal port~$u$ successfully. %

Summarizing, the overall algorithm for augmenting from a horizontal
port $u$ now works as follows.  By exploiting
Lemmas~\ref{lem:label_first_candidate}
and~\ref{lem:label_any_candidate}, we test the candidates in the order
as they appear on $f$ until we find the first candidate $e$ for which
$\Gamma^u_{e}$ does not contain a decreasing cycle.  Using
Fact~\ref{fact:consecutive-done} we either find that $\Gamma^{u}_{e}$
is valid, or we find a horizontal path as described above. In both
cases this allows us to determine an edge whose insertion does not
introduce a monotone cycle.  Since in each test for a decreasing
cycle the edge $uz$ can be restricted to have label~$0$, each of the
tests takes linear time.  This improves the running time of the
rectangulation algorithm to $O(n^{3})$.

Instead of linearly searching for a suitable candidate for $u$ we can
employ a binary search on the candidates, which reduces the number of
validity tests for $u$ from linear to logarithmic. To do this
efficiently we first compute the list of all candidates
$e_1,\dots,e_k$ for $u$ in time linear to the size of $f$. Next, we
test if the augmentation $\Gamma^u_{e_1}$ is valid. If it is, we are
done.

Otherwise, we start the binary search on the list $e_1,\dots,e_k$, 
where $k$ is the number of candidates for $u$.
The search maintains a sublist $e_i,\dots,e_j$ of consecutive candidates such 
that $\Gamma^u_{e_i}$ contains a decreasing cycle and $\Gamma^u_{e_j}$ does 
not. Note that this invariant holds in the beginning because we explicitly test 
for a decreasing cycle in $\Gamma^u_{e_1}$ and there is no decreasing cycle in 
$\Gamma^u_{e_{k}}$ by Fact~\ref{fact:first-last}.
If the list consists of only two consecutive candidates, i.e., $j=i+1$, we stop.
Otherwise, we set $m=\lfloor (i+j)/2\rfloor$ and test if $\Gamma^u_{e_m}$ 
contains a decreasing cycle. If it does, we recurse on $e_m,\dots,e_j$ and 
otherwise on $e_i,\dots,e_m$.
As the invariant is preserved we end up with two consecutive candidates $e_i$ 
and $e_{i+1}$ such that $\Gamma^u_{e_i}$ contains a decreasing cycle and 
$\Gamma^u_{e_{i+1}}$ does not.
In this situation Fact~\ref{fact:consecutive-done} guarantees that 
we find a valid augmentation.

Note that this augmentation may be different from the one we obtain if 
we test all candidates sequentially since there might be a candidate with a 
valid augmentation between two candidates whose augmentations contain 
decreasing cycles.

\begin{lemma}\label{lem:binary_search}
  Using binary search we find a valid augmentation for $u$ in $O(n\log n)$ 
  time.
\end{lemma}
\begin{proof}
 If the augmentation to the first candidate does not contain a decreasing 
 cycle, it is valid by Fact~\ref{fact:first-last}, and we are done. Otherwise, 
 the invariant that the augmentation for the first candidate in the list 
 contains a decreasing cycle and the augmentation for the last candidate does 
 not guarantee that we end up in a situation where we find a valid 
 augmentation 
 by Fact~\ref{fact:consecutive-done}. This establishes the correctness of the 
 augmentation algorithm based on binary search.
  
 Applying Lemma~\ref{lem:label_first_candidate} testing the first candidate 
 requires $O(n)$ time. If this augmentation contains a decreasing cycle, 
 Lemma~\ref{lem:label_any_candidate} guarantees that all other tests for 
 decreasing cycles can be implemented in $\O(n)$ time as well. The final test 
 for an increasing cycle can be replaced by a test for a horizontal path by 
 Fact~\ref{fact:consecutive-done}. In total, there are $\O(\log n)$ 
 tests with a total running time of $\O(n\log n)$. 
\end{proof}

Since there are at most $n$ ports to remove, we obtain that any
$4$-planar graph with valid ortho-radial representation can be
rectangulated in $O(n^2\log n)$ time.  Using Corollary 19
from~\cite{bnrw-ttsmford-17-arxiv} this further implies that a
corresponding ortho-radial drawing can be computed in $O(n^2\log n)$
time.

\begin{restatable}{restatable-theorem}{augmentationbinarysearch}
  \label{thm:augmentation-binary-search}
  Given a valid ortho-radial representation $\Gamma$ of a graph $G$, a
  corresponding rectangulation can be computed in $O(n^2 \log n)$
  time.
\end{restatable}

\subsection{2nd Improvement -- 2-Phase Augmentation Step}
\label{sec:2-phase-appendix}

In this section we describe an improvement of our algorithm that
reduces the total number of validity tests to $O(n)$. Hence, with this
improvement the running time of our algorithm is $O(n^2)$. Since the
construction is rather technical, we first present a high-level
overview in Section~\ref{sec:2-phase-appendix-intuition}. Afterwards we
present all technical details and formal proofs in
Section~\ref{sec:2-phase-appendix-details}.

\subsubsection{High-Level Overview of 2nd
  Improvement}\label{sec:2-phase-appendix-intuition}

In order to reduce the total number of validity tests to
$O(n)$, we add a second phase to our augmentation step that
post-processes the resulting augmentation after each step. More
precisely, the first phase of the augmentation step inserts a new edge
$uz$ in a given ortho-radial representation $\Gamma$ for a port $u$ as
before; we denote the resulting valid ortho-radial representation by
$\Gamma'$.

Afterwards, if $u$ is a horizontal port, we apply the second phase on
$u$. Let $e_1,\dots,e_k$ be the candidates of $u$, where $e_k$ is the
candidate for the first validity test that does not fail in the first
phase.  We call $e_1$, $e_{k-1}$ and $e_k$ \emph{boundary candidates}
and the others \emph{intermediate candidates}.
  The second phase
augments $\Gamma'$ such that afterwards each intermediate candidate belongs to a
rectangle in the resulting ortho-radial
representation~$\Gamma''$. Further, $\Gamma''$ has fewer vertices with
concave angles becoming horizontal ports and at most two more
vertices with concave angles becoming vertical ports during the remaining augmentation process. Since the second
phase is skipped for vertical ports, $O(n)$ augmentation steps are
executed overall.  Moreover, each edge can be an intermediate
candidate for at most one vertex, which yields that there are $O(n)$
intermediate candidates over all augmentation steps.  Finally, for
each port there are at most three boundary candidates, which yields
$O(n)$ boundary candidates over all augmentation steps. Assigning the
validity tests to their candidates, we conclude that the algorithm
executes $O(n)$ validity tests overall. Altogether, we obtain $O(n^2)$
running time for our algorithm.

We briefly sketch the concepts of the second phase. We make use of the
following lemma, which follows from Lemma~13
in~\cite{bnrw-ttsmford-17-arxiv}.

\begin{restatable}{restatable-lemma}{increasingdecreasingdisjoint}
  \label{lem:increasing_decreasing_disjoint}
  A decreasing and an increasing cycle do not have any common vertex.
\end{restatable}

\begin{proof}
Let $C_1$ be an increasing and $C_2$ a decreasing cycle. Assume that they have 
a common vertex. But then there also is a common vertex on the central face~$g$ 
of the subgraph $C_1+C_2$.
Consider any maximal common path~$P$ of $C_1$ and $C_2$ on $g$. We denote the 
start vertex of $P$ by $v$ and the end vertex by $w$. Note that $v$ may equal 
$w$. By Lemma~13 in~\cite{bnrw-ttsmford-17-arxiv}, the edge to $v$ on the 
decreasing cycle $C_2$ lies strictly in the exterior of $C_1$. Similarly, the 
edge from $w$ on $C_2$ lies strictly in the interior of $C_1$. Hence, 
$\subpath{C_2}{w,v}$ crosses $C_1$. Let $x$ be the first intersection. But the 
edge to $x$ on $C_2$ lies strictly in the interior of $C_1$ contradicting 
Lemma~13 in~\cite{bnrw-ttsmford-17-arxiv}.
\end{proof}

Hence, if an edge of an increasing cycle~$C$ lies in the interior of a
decreasing cycle~$C'$ and in the exterior of another decreasing
cycle~$C''$, then all edges of $C$ lie in the interior of $C'$ and in
the exterior of $C''$. We say that $C'$ and $C''$ \emph{wedge}
$C$. We use that observation as follows.

Let $\Gamma$ and $\Gamma_0$ be the ortho-radial representations before
and after the first phase of the augmentation step,
respectively. Further, as defined above let $e_1,\dots,e_k$ be the
candidates of the horizontal port $u$ considered in the augmentation
step.  We first note that there are decreasing cycles $C_1$ and $C_2$
in $\Gamma^u_{e_{1}}$ and $\Gamma^u_{e_{k-1}}$, respectively. We
simulate these cycles in $\Gamma_0$ as follows. We replace the edge
$uz$ inserted in the first phase by a structure $K$ that consists of
the paths $R$, $T$ and $B$ as illustrated in
Figure~\ref{fig:second-phase:step1}. The exact definition of $K$
relies on whether $uz$ belongs to a \emph{horizontal cycle} or not; we
call an essential cycle \emph{horizontal} if it only consists of
horizontal edges.

Yet,
in both cases $T$ contains a vertical edge $e_\T$ and is connected to
a subdivision vertex $t$ on $e_1$. Analogously, $B$ contains a
vertical edge $e_\B$ and is connected to a subdivision vertex $b$ on
$e_{k-1}$.  Let $\Gamma_1$ be the resulting ortho-radial
representation.  We show that it is valid. Afterwards, we use $K$ to 
simulate $C_1$ and $C_2$. More precisely, there is an essential
cycle~$C_\T$ in $\Gamma_{1}$ that contains $C_1[t,u]$, a part of $R$,
and $T$. Furthermore, there is an essential cycle $C_\B$ in
$\Gamma_{1}$ that contains $C_2[b,u]$, a part of $R$ and the path
$B$. We show that $C_\T$ has a negative label on $e_\T$ and
non-negative labels on all other edges. Similarly, we show that $C_\B$
has a negative label on $e_\B$ and non-negative labels on all other
edges.  Hence, apart from $e_\T$ and $e_\B$, both $C_\T$ and $C_\B$
behave as if they were decreasing cycles.

Consider the face $f_1$ that locally lies to the left of $B$ and to
the right of $T$. All intermediate candidates of $u$ lie on $f_1$. We
rectangulate $f_1$ as follows. We connect each vertical port to its
first candidate, which yields a valid ortho-radial representation by
Fact~\ref{fact:vertical-candidate}. Further, we connect each horizontal port to its
last candidate. By Fact~\ref{fact:first-last} this may produce
increasing cycles, but no decreasing cycles. We argue that any such
increasing cycle~$C$ is wedged by $C_\T$ and $C_\B$ and neither shares
vertices with $C_\T$ nor with $C_\B$. Since $C_\T$ and $C_\B$ share
vertices and the newly introduced edge of $C$ lies between $C_\T$ and
$C_\B$, the cycle $C$ cannot exist. Thus, rectangulating $f_1$ yields a
valid ortho-radial representation~$\Gamma_2$. In particular, since all
intermediate candidates of $u$ lie on $f_1$, they lie on rectangles
afterwards. Finally, we rectangulate the face $f_2$ that locally lies
to the left of $T$. Since $f_2$ has constant size, we can do this in
$O(n)$ time using the original augmentation step without the
second phase. By doing this, we resolve all horizontal ports that we
have introduced by inserting $K$. In
Section~\ref{sec:2-phase-appendix} we argue that the new augmentation step needs $O(n)$
time amortized over all augmentation steps. Altogether, we obtain our
main result.

\begin{restatable}{restatable-theorem}{nsquaredaugmentation}
  \label{thm:fast-augmentation}
  Given a valid ortho-radial representation $\Gamma$ of a graph $G$, a
  corresponding rectangulation can be computed in $O(n^2)$
  time.
\end{restatable}

\begin{figure}[t]
    \centering
    \begin{subfigure}[b]{0.49\textwidth}
      \centering
      \includegraphics[page=4,width=\textwidth]{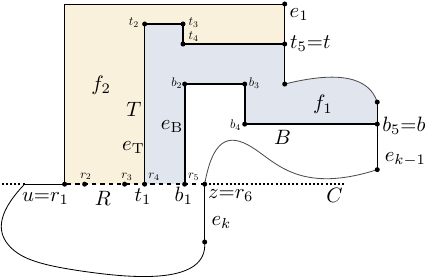}
      \subcaption{$\Gamma_1$: $r_2r_3$ points right.}
      \label{fig:second-phase:step1:right}
    \end{subfigure}
    \begin{subfigure}[b]{0.49\textwidth}
      \centering
      \includegraphics[page=2,width=\textwidth]{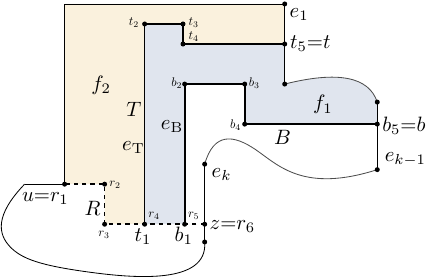}
      \subcaption{$\Gamma_1$: $r_2r_3$ points downwards.}
      \label{fig:second-phase:step1:downwards}
    \end{subfigure} 
    \caption{Illustration of Step 1, which inserts $R$, $T$ and $B$
      into $\Gamma_0$. Depending on whether $uz$ lies on a
      horizontal cycle~$C$ in $\Gamma_0$, the edge $r_2r_3$ points
      (\subref{fig:second-phase:step1:right}) to the right or
      (\subref{fig:second-phase:step1:downwards}) downwards.}
    \label{fig:second-phase:step1}
\end{figure}

\subsubsection{Details of 2nd Improvement}\label{sec:2-phase-appendix-details}
We now describe the second phase of the augmentation step in greater
detail.  Let $u$ be the port that is currently considered by the
augmentation step. In case that $u$ is a vertical port, we skip the
second phase.  So assume that $u$ is a horizontal port and that the
first phase of the augmentation step inserts an edge $uz$ that points
to the right; the case that $uz$ points to the left can be handled
analogously by flipping the cylinder. We denote the resulting valid
ortho-radial representation by $\Gamma_0$. Let $f$ be the face that
locally lies to the left of $uz$. Further, let $e_1,\dots,e_k$ be the
candidate edges of $u$ that were considered before inserting $uz$,
that is, the validity test failed for $e_1,\dots,e_{k-1}$ and
succeeded for $e_k$.  If $k < 4$ there are no intermediate candidates,
and we skip the second phase. So assume that $k\geq 4$. We apply the
following three steps to obtain the ortho-radial representations
$\Gamma_i$ with $1\leq i\leq 3$. Later on, we show that each
representation~$\Gamma_i$ is valid, all intermediate candidates lie on
rectangles in $\Gamma_3$, and there are only two more vertices
becoming vertical ports and no more vertices becoming horizontal ports
in $\Gamma_3$ than in $\Gamma_0$.

\textit{Step 1}. We replace $uz$ by a structure consisting of three
paths $R$, $T$, and $B$ as follows; see
Figure~\ref{fig:second-phase:step1} for an illustration.  The path $R$
connects $u$ with $z$. It consists of six vertices $r_i$ with
$1\leq i\leq 6$ such that $r_1=u$ and $r_6=z$. Apart from $r_2r_3$ all
edges point to the right. For the direction of $r_2r_3$ we distinguish
two cases.  If $uz$ lies on a cycle whose labels are all $0$, the edge
$r_2r_3$ points to the right and otherwise $r_2r_3$ points downwards;
see
Figure~\ref{fig:second-phase:step1}(\subref{fig:second-phase:step1:right})
and (\subref{fig:second-phase:step1:downwards}), respectively.

The path $T$ consists of five vertices $t_i$ with $1\leq i \leq 5$
such that $t_1=r_4$ and $t_5$ subdivides~$e_1$. The edge $t_1t_2$
points upwards, the edge $t_3t_4$ points downwards, and the other two
edges point to the right. Similarly, the path $B$ consists of five
vertices $b_i$ with $1\leq i \leq 5$ such that~$b_1=r_5$ and $b_5$
subdivides the edge $e_{k-1}$. Further, the edge $b_1b_2$ points
upwards, the edge~$b_3b_4$ points downwards, and the other two edges
point to the right.

We denote the resulting ortho-radial representation by
$\Gamma_1$. Further, let $f_1$ be the face that locally lies to the
right of $T$, and let $f_2$ be the face that locally lies to the left of $T$.

\textit{Step 2.}  We iteratively resolve the ports in $f_1$ until the
face is rectangulated. To that end, let $\Pi$ be the ortho-radial
representation of the previous iteration; we start with
$\Pi=\Gamma_1$. Further, let $u'$ be the currently considered port of
$f_1$ and let $e'_1,\dots,e'_l$ be its candidates. If $u'$ is a
vertical port, we take $\Pi^{u'}_{e'_1}$ and otherwise
$\Pi^{u'}_{e'_l}$ as result of the current iteration. The procedure
stops when $f_1$ is completely rectangulated. We denote the resulting
ortho-radial representation by~$\Gamma_2$.

\textit{Step 3.} Starting with $\Gamma_2$, we rectangulate the face
$f_2$ by iteratively applying the augmentation step without
Phase 2 until there are no ports left in $f_2$. We denote
the resulting ortho-radial representation by $\Gamma_3$.

\paragraph*{Correctness.}
We now prove that the second phase yields a valid ortho-radial
representation~$\Gamma_3$ by showing that each step yields a valid
ortho-radial representation. We use the same notation as above.

\textit{Step 1.} In order to show the correctness, we successively add
the paths $R$, $T$ and $B$ to $\Gamma_0$ and prove the validity
of each created ortho-radial representation. To that end, let
$\Gamma_R=\Gamma_0-uz+R$,
 $\Gamma_T=\Gamma_R+T$
and $\Gamma_B=\Gamma_T+B=\Gamma_1$.

\begin{lemma}\label{lem:second-phase:no-monotone-cycle-on-R}
  The ortho-radial representation $\Gamma_R$ is valid.
\end{lemma}

\begin{proof}
  Assume that $\Gamma_R$ contains a monotone cycle~$C$. Since
  $\Gamma_0$ is valid, this cycle uses $R$.  In case that the edge
  $r_2r_3$ of $R$ points to the right, we can interpret $R$ as a
  single edge on $C$, because the labels of $C$ on $R$ are
  identical. Hence, $C$ corresponds to a cycle~$C'=C[z,u]+uz$ in
  $\Gamma_0$, where $uz$ is subdivided by some additional vertices on
  $C$. Thus, $C$ and $C'$ have the same labels, which contradicts
  that~$\Gamma_0$ is valid.

  So assume that $r_2r_3$ points downwards. Without loss of
  generality, we assume that $C$ uses $uz$; the case that $C$ uses
  $zu$ can be handled identically. By construction the vertex $z$ is a
  newly introduced vertex subdividing $e_k$. Since $e_k$ is vertical,
  this implies that apart from $r_2r_3$ the cycle $C$ contains another
  vertical edge on $C[r_6,r_1]$. Further, $C[r_6,r_1]$ is also
  contained in $\Gamma_0$. Hence, $C[r_6,r_1]+uz$ forms an essential
  cycle~$C'$ in $\Gamma_0$ with at least one vertical edge.

  We show that for any common edge of $C$ and $C'$ the labels of $C$
  and $C'$ are identical. Since all vertical edges of $C'$ also belong
  to $C$, and $C'$ has at least one vertical edge, this shows that $C'$
  is also monotone, which contradicts the validity of $\Gamma_0$.

  Let $P$ be an elementary path to $C'$ and let $v$ be the end vertex
  of $P$ on $C'$. Since $uz$ belongs to $C'$, the path does not
  contain this edge. Hence, $P$ is also contained in $\Gamma_R$ and it is
  an elementary path for $C$.  Now consider an edge $e$ that belongs
  to both $C$ and $C'$. Let $Q$ be the path on $C$ from $v$ to the target
  of $e$ and, analogously, let $Q'$ be the path on $C'$ from $v$ to
  the target of $e$. The path $Q$ contains $R$ if and only if $Q'$
  contains $uz$. Hence, if $Q$ does not contain $R$, both paths are
  identical, and we obtain
  \[
    \ell_C(e)=\rot(e^\star+P+Q)=\rot(e^\star+P+Q')=\ell_{C'}(e).\]
  So assume that $Q$ contains $R$. By construction we have
  $\rot(au+uz+zb)=\rot(au+R+zb)$, where $a$ is the vertex on $P+Q$
  before $u$ and $b$ is the vertex on $Q$ after $z$. It holds
  \begin{align*}    
  \ell_C(e)&=\rot(e^\star+P+Q[v,u])+\rot(au+R+zb)+\rot(Q[zb,e])\\
  &=\rot(e^\star+P+Q'[v,u])+\rot(au+uz+zb)+\rot(Q'[zb,e])=\ell_{C'}(e).
  \end{align*}
  Hence, for common edges the labels of $C$ and $C'$ are identical,
  which contradicts that $\Gamma_0$ is valid.
\end{proof}

Next, we prove that $\Gamma_T$ is valid. To
that end, we introduce the following definition. A \emph{cascading
  cycle} is a non-monotone essential cycle that can be partitioned
into two paths $P$ and $Q$ such that the labels on $P$ are $-1$
and the labels on $Q$ are non-negative.  We further require that the edges 
incident to the internal vertices of $P$ either all lie in the interior of $C$ 
or they all lie in the exterior of $C$. In the first case we call $C$ an 
\emph{outer} cascading cycle and in the second case an \emph{inner} cascading 
cycle. The path $P$ is the \emph{negative path} of the cycle.

To show that $\Gamma_T$ is valid, we construct a cascading cycle $C_\T$ in 
$\Gamma_T$ as follows. Let $C_1$ be the
outermost decreasing cycle in $\Gamma^{u}_{e_{1}}$ and let $ut_5$ be the newly 
inserted edge in
$\Gamma^{u}_{e_{1}}$. We replace $ut_5$ by $R[u,r_4]+T$ obtaining the
cycle $C_\T$, which is well-defined because $C_1$ uses $ut_5$ in that
direction by Lemma~\ref{lem:label_first_candidate}.

\begin{lemma}\label{lem:second-phase:cascading-cycles-CT}
  $C_\T$ is a cascading cycle no matter whether $r_2r_3$ points to the
  right or downwards. In particular, $t_1t_2$ is the negative path of
  $C_\T$.
\end{lemma}

\begin{proof}
  Let $C_1$ be the outermost decreasing cycle that in $\Gamma^{u}_{e_{1}}$ in 
  the first phase. There is an
  elementary path $P$ from $e^\star$ that ends at a vertex $v$ on
  $C$. Since $P$ does not contain $ut_5$ in either
  direction, it is also an elementary path for $C_\T$ in
  $\Gamma_T$. Let further $Q$ be the path from $v$ to $u$. Since
  $C_1$ uses $ut_5$ in that direction, the path $Q$ does not use
  $ut_5$. This implies that $Q$ also exists on $C_\T$. Thus, $ut_5$
  and $ur_2$ have the same label on $C_1$ and $C_T$. By
  Lemma~\ref{lem:label_first_candidate} the edge $ut_5$ has label $0$
  on $C_1$. If $r_2r_3$ points to the right the sequence of the labels on
  $R[u,r_4]+T$ is therefore $0$, $0$, $0$, $-1$, $0$, $1$, $0$. If $r_2r_3$
  points downwards the sequence is $0$, $1$, $0$, $-1$, $0$, $1$,
  $0$. In both cases $C_\T$ is not monotone.

  We now show that $t_1t_2$ is the only edge on $C_\T$ with negative
  label, which shows that $C_\T$ is a cascading cycle. In particular, the 
  negative path consists of only one edge and therefore it has no internal 
  vertices. We
  observe that $ut_5$ and $t_4t_5$ have the same label. Hence, for any
  common edge~$e$ of $C_1$ and $C_\T$ there are paths~$Q$ and $Q'$
  from $v$ to $e$, respectively, such that
  $\rot(e^\star+P+Q+e)=\rot(e^\star+P+Q'+e)$. This implies that for
  any common edge of $C_1$ and $C_\T$, the labels of both cycles are
  identical. Since $C_1$ is a decreasing cycle, all common edges have
  non-negative labels. Altogether, $t_1t_2$ is the only edge of $C_\T$
  with a negative label.
\end{proof}

Using this lemma we prove
that there is no decreasing cycle in $\Gamma_T$.

\begin{lemma}\label{lem:second-phase:no-decreasing-cycle-on-T}
  There is no decreasing cycle in $\Gamma_T$.
\end{lemma}

\begin{proof}
  Assume that $T$ is contained in a decreasing cycle $C$ (in either
  direction).  Let $H=C+C_\T$ be the common sub-graph of $C$ and
  $C_\T$, and let $g$ be the central face of $H$. We distinguish the
  following two cases.

  \textit{Case 1, $T$ is part of $g$.} 
  First assume that $C_\T$ and
  $C$ use $T$ in opposite directions. Since the central face locally
  lies to the right of any essential cycle, this implies that the
  central face lies to the left and right of $T$. Consequently, the
  central face is not simple, which contradicts that $H$ is biconnected. So 
  assume that $C$ and $C_\T$ use $T$ in the
  same direction. By
  Proposition~\ref{lem:repr:equal_labels_at_intersection} it holds
  $\ell_{C_\T}(t_1t_2)=\ell_C(t_1t_2)$. Since the cycle $C_\T$ is a
  cascading cycle with negative path $t_1t_2$ by
  Lemma~\ref{lem:second-phase:cascading-cycles-CT}, it is
  $\ell_{C_\T}(t_1t_2)=-1$. Thus, $C$ is not a decreasing cycle.

  \textit{Case 2, $T$ is not part of $g$.} Let $C'$ be the essential
  cycle formed by $g$. Since $C'$ consists of edges of $C$ and $C_\T$,
  the corresponding labels of $C$ and $C_\T$ also apply on $C'$ by
  Proposition~\ref{lem:repr:equal_labels_at_intersection}. Further,
  since $C$ is a decreasing cycle and $T$ is the only part of $C_\T$
  that has a negative label on $C_\T$ by
  Lemma~\ref{lem:second-phase:cascading-cycles-CT}, the cycle $C'$
  only has non-negative labels. Since $T$ does not lie on $C'$ but on $C_\T$,  
  $C$ has at least one vertex with $C'$ in
  common. This implies that $C'$ has at least one positive label,
  because otherwise $C$ could not be a decreasing cycle by
  Proposition~\ref{lem:rect:two_cycles_horizontal}. Altogether, $C'$
  is a decreasing cycle that also exists in $\Gamma_R$, which
  contradicts its validity.
\end{proof}

To show that $\Gamma_T$ contains no increasing cycle, we introduce a general 
lemma about the interaction of cascading and increasing cycles.

\begin{lemma}\label{lem:cascading_and_increasing_cycles}
  Let $C$ be a cascading cycle and $C'$ an increasing cycle. Either $C$ lies in 
  the interior of $C'$ or vice versa.
\end{lemma}
\begin{proof}
  We assume without loss of generality that $C$ is an outer cascading cycle. 
  The 
  case that it is an inner cascading cycle can be handled by flipping the 
  cylinder, which exchanges the exterior and interior of essential cycles but 
  keeps the labels.
  
  \begin{figure}
    \centering
    \begin{minipage}[b]{0.39\textwidth}
      \begin{subfigure}[b]{0.9\textwidth}
        \centering
        \includegraphics[]{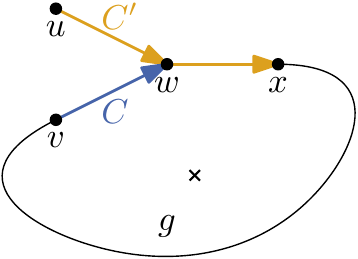}
        \subcaption{The central face $g$.}
        \label{fig:cascading_increasing-g}
      \end{subfigure}    
    \end{minipage}
    \begin{minipage}[b]{0.60\textwidth}
      \begin{subfigure}[b]{0.49\textwidth}
        \centering
        \includegraphics[page=1]{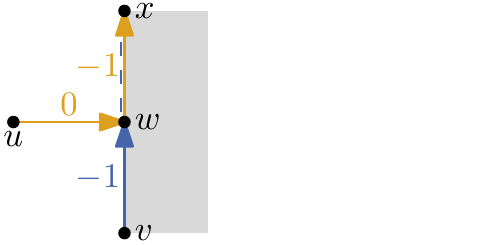}
        \subcaption{The labels are $-1$. }
        \label{fig:cascading_increasing-vertical}
      \end{subfigure}    
      \begin{subfigure}[b]{0.49\textwidth}
        \centering
        \includegraphics[page=2]{figures/cascading_increasing_cases.pdf}
        \subcaption{The labels are $0$. }
        \label{fig:cascading_increasing-horizontal}
      \end{subfigure}
      
      \begin{subfigure}[b]{\textwidth}
        \centering
        \includegraphics[page=3]{figures/cascading_increasing_cases.pdf}
        \subcaption{$\ell_C(vw) = -1$ and $\ell_{C'}(wx)=0$. }
        \label{fig:cascading_increasing-turn}
      \end{subfigure}
    \end{minipage}
    \caption{A common vertex $w$ on the central face $g$ of the subgraph formed 
      by the cascading cycle $C$ and the increasing cycle $C'$ and possible 
      labels 
      of the edges incident to $w$.}
  \end{figure}
  
  Let $g$ be the central face of the subgraph formed by the cycles $C$ and 
  $C'$. 
  If $g$ is neither $C$ nor $C'$, there are edges $vw$ and $wx$ on $g$ such 
  that 
  $vw$ lies on $C$ but not $C'$, and $wx$ lies on $C'$; see 
  Figure~\ref{fig:cascading_increasing-g}. Let 
  $uw$ be the edge on $C'$ entering $w$. By construction, $uw$ lies strictly in 
  the exterior of $g$. Hence, $vwx$ cannot make a left turn at $w$ and therefore
  $\ell_g(vw) \le \ell_g(wx)$. Combining this with 
  Proposition~\ref{lem:repr:equal_labels_at_intersection} and the bounds for the
  labels on $C$ and $C'$ we get
  \begin{equation*}
  -1 \le \ell_C(vw) = \ell_g(vw) \le \ell_g(wx) = \ell_{C'}(wx) \le 0.
  \end{equation*}
  There are three cases for the labels $\ell_C(vw)$ and $\ell_{C'}(wx)$:
  Either both are $-1$, both are $0$, or $\ell_C(vw)=-1$ and $\ell_{C'}(wx)=0$.
  
  If both labels are $-1$, the edge on $C$ after $w$ is $wx$; see 
  Figure~\ref{fig:cascading_increasing-vertical}. It cannot be $wu$ 
  because then the label of $wu$ would be $-2$. Hence, $w$ is an internal 
  vertex 
  of the negative path of $C$. But $uw$ lies in the exterior of $C$ contradicting that 
  $C$ is an outer cascading cycle.
  
  If both labels are $0$, the edge $uw$ must point down and therefore 
  $\ell_{C'}(uw)=1$, which contradicts that $C'$ is increasing; see 
  Figure~\ref{fig:cascading_increasing-horizontal}.
  
  Hence, $\ell_C(vw)=-1$ and $\ell_{C'}(wx) = 0$; see 
  Figure~\ref{fig:cascading_increasing-turn}. As before $uw$ cannot point 
  down, which implies that it points right. The edge after $w$ on $C$ does not 
  point left because it would have label $-2$. It does not point up since then 
  $w$ would be an internal vertex of the negative path, and we get a 
  contradiction as in the first case. Thus, $wx$ lies on $C$ and $w$ is the 
  endpoint of the negative path $P$ of $C$. Therefore, there is a common path 
  of 
  $C$ and $C'$ starting at $w$ and ending at a vertex $y$. Since $C$ is not 
  monotone, it has an edge with a positive label and hence $y$ does not lie on 
  $P$. Therefore, the edge on 
  $C$ after $y$ has a non-negative label and the edge~$yz$ on $C'$ after $y$ 
  has 
  a non-positive label. By Lemma~13 in Reference~\cite{bnrw-ttsmford-17-arxiv}, 
  the edge $yz$ lies in the exterior of $C$.
  In total, this shows that no part of $C'$ lies strictly in the interior of 
  $C$ and 
  therefore $g=C$.
\end{proof}

Applying this lemma to the situation of $C_\T$ we prove that $\Gamma_T$ does 
not contain any increasing cycles. Together with 
Lemma~\ref{lem:second-phase:no-decreasing-cycle-on-T} this yields that 
$\Gamma_T$ is valid.

\begin{lemma}\label{lem:second-phase:no-increasing-cycle-on-T}
  There is no increasing cycle in $\Gamma_T$.
\end{lemma}

\begin{proof}
  Assume that $\Gamma_T$ contains an increasing cycle $C$, which uses $T$ in 
  any direction. Lemma~\ref{lem:cascading_and_increasing_cycles} implies that 
  the central face $g$ of the subgraph formed by the two essential cycles 
  $C$ and $C_\T$ is either $C$ or $C_\T$.
  In particular, $T$ or $\reverse{T}$ lies on $g$. Hence, both $C$ and $C_\T$ 
  use $T$ in the same direction as otherwise $g$ would lie in the exterior of 
  one of these cycles. But this would contradict that they are essential.
  Hence, they both contain $T$ in this direction and $T$ also lies on $g$. 
  By Proposition~\ref{lem:repr:equal_labels_at_intersection} both
  cycles have the same labels on $T$. Since $\ell_{C_\T}(t_1t_2)=-1$ by 
  Lemma~\ref{lem:second-phase:cascading-cycles-CT},
  we obtain $\ell_C(t_3t_4)=\ell_{C_\T}(t_3t_4)=1$. 
  Consequently,
  $C$ is not an increasing cycle.  
\end{proof}

Hence, $\Gamma_T$ is valid. We analogously prove the validity of
$\Gamma_B$ as for $\Gamma_T$.
Let $C_2$ be the outermost decreasing cycle that
in $\Gamma^{u}_{e_{k-1}}$ and let $ub_5$
be the newly inserted edge in $\Gamma^{u}_{e_{k-1}}$. We replace
$ub_5$ by $R[u,r_5]+B$ obtaining the cycle $C_\B$, which is
well-defined because $C_2$ uses $ub_5$ in that direction by
Lemma~\ref{lem:label_first_candidate}.

\begin{lemma}\label{lem:second-phase:cascading-cycles-CB}
  $C_\B$ is a cascading cycles no matter whether $r_2r_3$ points to
  the right or downwards. In particular, $b_1b_2$ is the negative
  path of $C_\B$.
\end{lemma}

We omit the proof since it uses the same arguments as the proof
of~Lemma~\ref{lem:second-phase:cascading-cycles-CT}. Using similar arguments as 
in the proofs of Lemmas~\ref{lem:second-phase:no-decreasing-cycle-on-T} 
and~\ref{lem:second-phase:no-increasing-cycle-on-T}, we obtain that $\Gamma_B$ 
is valid.

\begin{lemma}\label{lem:second-phase:no-monotone-cycle-on-B}
 The ortho-radial representation $\Gamma_B=\Gamma_1$ is valid.
\end{lemma}

\textit{Step 2.} By
Lemma~\ref{lem:second-phase:no-monotone-cycle-on-B} the ortho-radial
representation $\Gamma_1$ of \textit{Step 1} is valid. We now prove
that $\Gamma_2$ is a valid ortho-radial representation. We use the
same notation as in the description of the algorithm.

Starting with the valid ortho-radial representation~$\Gamma_1$, the
procedure iteratively resolves ports in the face $f_1$, which locally
lies to the right of $T$. In case that we resolve a vertical port $u'$
in a representation $\Pi$, the resulting ortho-radial
representation~$\Pi^{u'}_{e'_1}$ is valid by
Fact~\ref{fact:vertical-candidate}, where $e'_1$ is the first
candidate of $u'$.  So assume that $u'$ is a horizontal port. In that
case we take $\Pi^{u'}_{e'_l}$ for the next iteration, where $e'_l$ is
the last candidate of $u'$. We observe that the augmentation of $f_1$
may subdivide edges on the negative paths of $C_\T$ and $C_\B$, but
the added edges lie in the interior of $C_\T$ and the exterior of
$C_\B$. Hence, $C_\T$ remains an outer cascading cycle and $C_\B$ an
inner cascading cycle.

\begin{lemma}
 The ortho-radial representation $\Pi^{u'}_{e'_l}$ is valid. 
\end{lemma}

\begin{proof}
  Assume that $\Pi^{u'}_{e'_l}$ is not valid. Hence, there is a
  monotone cycle $C$ that uses $e=u'z'$, where $z'$ is the vertex
  subdividing $e'_l$. Since $e'_l$ is the last candidate of $u'$, the
  cycle $C$ is increasing by Fact~\ref{fact:first-last}.
  By construction $e$ strictly lies in the interior of $C_\T$ and the exterior 
  of $C_\B$. This implies that $C$ lies in the interior of $C_\T$ and the 
  exterior of $C_\B$ by Lemma~\ref{lem:cascading_and_increasing_cycles}.
  In other words, $C$ is contained in the subgraph $H$ formed by the 
  intersection of the interior of $C_\T$ and the exterior of $C_\B$. As 
  $\subpath{R}{r_1,r_4}$ belongs to both $C_\T$ and $C_\B$, it is incident to 
  the outer and the central face of $H$. Hence, removing $\subpath{R}{r_1,r_4}$ 
  leaves a subgraph without essential cycles. Thus, the essential cycle $C$ 
  includes $\subpath{R}{r_1,r_4}$.
  
  By Proposition~\ref{lem:repr:equal_labels_at_intersection} the labels of $C$ 
  and $C_\T$ are the same on $\subpath{R}{r_1,r_4}$. If $r_2r_3$ points 
  downwards, its label is $1$, which contradicts that $C$ is increasing.
  If otherwise $r_2r_3$ points right, it lies on an essential cycle, where all 
  labels are $0$. But then $C$ is not increasing by 
  Proposition~\ref{lem:rect:two_cycles_horizontal}.
\end{proof}

Altogether, applying the lemma inductively on the inserted edges, we
obtain that $\Gamma_2$ is valid.

\textit{Step 3.}  As we only apply the first phase of the augmentation
step on $f_2$, the resulting ortho-radial
representation~$\Gamma_3$ is also valid due to the correctness of the
first phase. This concludes the correctness proof of the second phase.
\begin{lemma}\label{lem:second-phase:correctness}
  The second phase produces a valid ortho-radial representation
  $\Gamma_3$ such that all intermediate candidates of $u$ lie on
  rectangles in $\Gamma_3$, and there are only two more vertices
  becoming vertical ports and no more vertices becoming horizontal
  ports in $\Gamma_3$ than in $\Gamma_0$.
\end{lemma}

\paragraph*{Running Time.}
We now prove that the rectangulation algorithm has $O(n^2)$ running
time in total. We first prove that the resulting ortho-radial
representation has $O(n)$ vertices and edges, which implies that
$O(n)$ augmentation steps are executed. Afterwards we show that the
algorithm spends $O(n^2)$ time in total for executing all augmentation
steps.

Consider a single augmentation step that resolves a horizontal
port~$u$ of a face $f$. Let $K=B+T+R$ be the construction that is
inserted during the second phase, and, furthermore, let $f_1$ and
$f_2$ be the faces as defined above.  After the second phase, only
$r_3$ and $b_4$ of the newly inserted vertices have concave angles;
all other concave angles of newly inserted vertices are resolved in
the second phase by rectangulating $f_1$ and $f_2$. By construction
both $r_3$ and $b_4$ can become vertical but not horizontal ports
during the remaining procedure. Hence, we insert the construction $K$
only for vertices that already have existed in the input
instance. Moreover, the rectangulation algorithm considers $O(n)$
vertical ports in total. Hence, the algorithm yields an ortho-radial
representation with $O(n)$ vertices and edges.  This also implies that
$O(n)$ augmentation steps are executed.

In the remainder we show that the algorithm invests $O(n^2)$ running
time in total for the execution of all augmentation steps. In
particular, we argue that the algorithm needs $O(n^2)$ time for the
computation of the candidates of all considered ports and all applied validity
tests. Since the first and second phase of the augmentation step needs
$O(1)$ time without considering the time necessary for the validity
tests and the computation of the candidates, we finally obtain that
the algorithm runs in $O(n^2)$ time.

Since the rectangulation algorithm yields an ortho-radial
representation with $O(n)$ vertices and edges, $O(n)$ ports are
resolved and $O(n)$ different candidate edges are considered. Since
the algorithm computes for each port its candidates only once (namely
when the port is resolved), the algorithm spends $O(n^2)$ time in
total to compute all candidate edges of the ports.

We now bound the number of applied validity tests. Recall that we
only apply validity tests in the first phase of the augmentation step
and when rectangulating the face $f_2$ in the second phase. By
Lemma~\ref{lem:second-phase:correctness} each edge can be an
intermediate candidate for at most one vertex, which yields that there
are $O(n)$ intermediate candidates over all augmentation steps.
Finally, for each vertex there are at most three boundary candidates,
which yields $O(n)$ boundary candidates over all augmentation
steps. Assigning the validity tests to their candidates, we conclude
that the algorithm executes $O(n)$ validity tests overall.
Altogether, we obtain $O(n^2)$ running time for the rectangulation
algorithm.

In particular, using Corollary~19 from~\cite{bnrw-ttsmford-17-arxiv},
given a graph $G$ with valid ortho-radial representation $\Gamma$, a
corresponding ortho-radial drawing $\Delta$ can be computed in
$O(n^{2})$ time.

\section{Conclusion}
\label{sec:conclusion}

In this paper, we have described an algorithm that checks the validity
of an ortho-radial representation in $O(n^{2})$ time. In the positive
case, we can also produce a corresponding drawing in the same running
time, whereas in the negative case we find a monotone cycle.  This
answers an open question of Barth et al.~\cite{bnrw-ttsmford-17} and
allows for a purely combinatorial treatment of the bend minimization
problem for ortho-radial drawings.
It is an interesting open question whether the running time can be
improved to near-linear.  However, our main open question is how to
find valid ortho-radial representations with few bends.

\bibliography{manuscript}

\appendix

\newpage

\end{document}